\documentclass[11  pt, oneside,onecolumn]{article} 
\usepackage[margin=1in,lmargin=1in,rmargin=1in,bmargin=1in,tmargin=1in]{geometry}  
\usepackage[latin1]{inputenc}
\usepackage{caption}
\usepackage{subcaption}
\usepackage{cite}

\usepackage{amsfonts}
\usepackage{amsmath}
\usepackage{mathrsfs}  
\usepackage{amssymb}
\usepackage{graphicx}
\usepackage{booktabs}
\usepackage{enumitem}
\usepackage{algorithm}
\usepackage{algorithmicx}
\usepackage{algpseudocode}
\algdef{SE}[DOWHILE]{Do}{doWhile}{\algorithmicdo}[1]{\algorithmicwhile\ #1}%

\newcommand{\sr}{\mathcal{S}_r}

\newcommand{\spagap}{\textsc{Spa-GAP}}

\usepackage{mathtools}
\DeclarePairedDelimiter\ceil{\lceil}{\rceil}
\DeclarePairedDelimiter\floor{\lfloor}{\rfloor}

\newcommand{\cut}{\textbf{cut}}
\newcommand{\dircut}{\textbf{dircut}}

\renewcommand{\vec}[1]{\boldsymbol{\mathrm{#1}}}

\DeclareMathOperator*{\minimize}{minimize}

\DeclareMathOperator{\argmin}{argmin}

\providecommand{\va}{\ensuremath{\vec{a}}}
\providecommand{\vb}{\ensuremath{\vec{b}}}

\providecommand{\vf}{\ensuremath{\vec{f}}}

\providecommand{\vw}{\ensuremath{\vec{w}}}
\providecommand{\vx}{\ensuremath{\vec{x}}}

\usepackage{xparse}

\usepackage{amsthm}
\newtheorem{theorem}{Theorem}
\numberwithin{theorem}{section}
\newtheorem{lemma}[theorem]{Lemma}

\newtheorem{observation}{Observation}
\numberwithin{observation}{section}
\newtheorem{definition}{Definition}
\numberwithin{definition}{section}

\numberwithin{conjecture}{section}

\newtheoremstyle{introstyle}
{}
{}
{\itshape}
{}
{\bfseries}
{.}
{ }
{}
\theoremstyle{mystyle}
\newtheorem*{theoremintro}{Theorem}

\theoremstyle{mystyle}
\newtheorem*{lemmaintro}{Lemma}

\usepackage{hyperref}
\hypersetup{
	colorlinks=true, 
	linktoc=all,     
	linkcolor=blue,  
	citecolor = red,
}

\author{Austin R.~Benson \\ Computer Science Dept. \\  Cornell University \\ arb@cs.cornell.edu \and Jon Kleinberg  \\ Computer Science Dept. \\ Cornell University\\ kleinberg@cornell.edu \and Nate Veldt \\Center for Applied Math  \\ Cornell University  \\ nveldt@cornell.edu
}

\begin{document}
\begin{titlepage}
\title{Augmented Sparsifiers for Generalized Hypergraph Cuts with Applications to Decomposable Submodular Function Minimization\thanks{This research was supported by NSF Award DMS-1830274, ARO Award W911NF19-1-0057, ARO MURI, JPMorgan Chase \& Co., a Simons Investigator Award, a Vannevar Bush Faculty Fellowship, and a grant from the AFOSR. The authors thank Pan Li for helpful conversations about decomposable submodular function minimization.}
}
\date{}
\maketitle
\begin{abstract}
	In recent years, hypergraph generalizations of many graph cut problems and algorithms have been introduced and analyzed as a way to better explore and understand complex systems and datasets characterized by multiway relationships. The standard cut function for a hypergraph $\mathcal{H} = (V,\mathcal{E})$ assigns the same penalty to a cut hyperedge, regardless of how its nodes are separated by a partition of $V$. Recent work in theoretical computer science and machine learning has made use of a generalized hypergraph cut function that can be defined by associating each hyperedge $e \in \mathcal{E}$ with a splitting function $\vw_e$, which assigns a (possibly different) penalty to each way of separating the nodes of $e$. When each $\vw_e$ is a \emph{submodular cardinality-based splitting function}, meaning that $\vw_e(S) = g(|S|)$ for some concave function $g$, previous work has shown that a generalized hypergraph cut problem can be reduced to a directed graph cut problem on an augmented node set. However, existing reduction procedures introduce up to $O(|e|^2)$ edges for a hyperedge $e$. This often results in a dense graph, even when the hypergraph is sparse, which leads to slow runtimes (in theory and practice) for algorithms that run on the reduced graph. 
	
	We introduce a new framework of sparsifying hypergraph-to-graph reductions, where a hypergraph cut defined by submodular cardinality-based splitting functions is $(1+\varepsilon)$-approximated by a cut on a directed graph. Our techniques are based on approximating concave functions using piecewise linear curves, and we show that they are optimal within an existing strategy for hypergraph reduction. We provide bounds on the number of edges needed to model different types of splitting functions. For $\varepsilon > 0$, in the worst case, we need $O(\varepsilon^{-1}|e| \log |e|)$ edges to reduce any hyperedge $e$, which leads to faster runtimes for approximately solving generalized hypergraph $s$-$t$ cut problems. For the common machine learning heuristic of a clique splitting function on a node set $e$, our approach requires only $O(|e|)$ nodes and $O(|e| \varepsilon^{-1/2} \log \log \frac{1}{\varepsilon})$ edges, instead of the $O(|e|^2)$ edges used with existing reductions. Equivalently, we can model the cut properties of a complete graph on $n$ nodes using $O(n)$ nodes and $O(n \varepsilon^{-1/2} \log \log \frac{1}{\varepsilon})$ directed and weighted edges. This sparsification leads to faster approximate min $s$-$t$ graph cut algorithms for certain classes of co-occurrence graphs that are represented implicitly by a collection of sets modeling co-occurrences. Finally, we apply our sparsification techniques to develop the first approximation algorithms for decomposable submodular function minimization in the case of cardinality-based component functions. 
	In addition to proving improved theoretical runtimes, we show in practice that our sparsification techniques lead to substantial runtime improvements for hypergraph clustering problems and benchmark image segmentation tasks for decomposable submodular function minimization.
\end{abstract}
\thispagestyle{empty}
\end{titlepage}


\section{Introduction}
Hypergraphs are a generalization of graphs in which nodes are organized into multiway relationships called hyperedges. Given a hypergraph $\mathcal{H} = (V,\mathcal{E})$ and a set of nodes $S \subseteq V$, a hyperedge $e \in E$ is said to be \emph{cut} by $S$ if both $S$ and $\bar{S} = V\backslash S$ contain at least one node from $e$. Developing efficient algorithms for cut problems in hypergraphs is an active area of research in theoretical computer science~\cite{checkuri2017computing,checkuri2018minimum,kogan2015sketching,fox2019minimum,chandrasekaran2018hypergraph}, and has been applied to problems in VLSI layout~\cite{alpert1995survey,hadley1995,karypis1999multilevel}, sparse matrix partitioning~\cite{akbudak2013cachelocality,Ballard:2016:HPS:3012407.3015144}, and machine learning~\cite{panli2017inhomogeneous,panli_submodular,veldt2020hyperlocal}.

Here, we consider recently introduced \emph{generalized} hypergraph cut functions~\cite{panli2017inhomogeneous,panli_submodular,veldt2020hypercuts,yoshida2019cheeger}, which assign different penalties to cut hyperedges based on how the nodes of a hyperedge are split into different sides of the bipartition induced by $S$. To define a generalized hypergraph cut function, each hyperedge $e \in E$ is first associated with a \emph{splitting function} $\vw_e\colon A \subseteq e \rightarrow \mathbb{R}^+$ that maps each node configuration of $e$ (defined by the subset $A \subseteq e$ in $S$) to a nonnegative penalty.
In order to mirror edge cut penalties in graphs, splitting functions are typically assumed to be
symmetric $(\vw_e(A) = \vw_e(e\backslash A)$) and only penalize cut hyperedges (i.e., $\vw_e(\emptyset) = 0$).
The generalized hypergraph cut function for a set $S \subseteq V$ is then given by
\begin{equation}
\label{hcf}
\cut_\mathcal{H}(S) = \sum_{e \in \mathcal{E}} \vw_e (S \cap e) \,.
\end{equation}

The standard hypergraph cut function is \emph{all-or-nothing}, meaning it assigns the same penalty to a cut hyperedge regardless of how its nodes are separated. Using the splitting function terminology, this means that $\vw_e(A) = 0$ if $A \in \{e ,\emptyset\}$, and $\vw_{e}(A) = w_e$ otherwise, where $w_e$ is a scalar hyperedge weight. 
One particularly relevant class of splitting function are submodular functions, which for all $A$,$B \subseteq e$ satisfy $\vw_e(A) + \vw_e(B) \geq \vw_e(A\cap B) + \vw_e(A \cup B)$. 
When all hyperedge splitting functions are submodular, solving generalized hypergraph cut problems is closely related to minimizing a decomposable submodular function~\cite{kolmogorov2012minimizing,nishihara2014convergence,panli2018revisiting,stobbe2010efficient,ene2017decomposable,ene2015random}, which in turn is closely related to energy minimization problems often encountered in computer vision~\cite{kohli2009robust,kolmogorov2004what,freedman2005energy}. The standard graph cut function is another well-known submodular special case of~\eqref{hcf}.

One of the most common techniques for solving hypergraph cut problems is to reduce the hypergraph to a graph sharing similar (or in some cases identical) cut properties. Arguably the most widely used reduction technique is clique expansion, which replaces each hyperedge with a (possibly weighted) clique~\cite{BensonGleichLeskovec2016,hadley1995,panli2017inhomogeneous,zien1999,Zhou2006learning}. 
In the unweighted case this corresponds to applying a splitting function of the form: $\vw_e(A) = |A| \cdot |e \backslash A|$.  Previous work has also explored other classes of submodular hypergraph cut functions that can be modeled as a graph cut problem on a potentially \emph{augmented} node set~\cite{freedman2005energy,kohli2009robust,kolmogorov2004what,lawler1973,veldt2020hypercuts}. This research primarily focuses on proving when such a reduction is possible, regardless of the number of edges and auxiliary nodes needed to realize the reduction. However, because hyperedges can be very large and splitting functions may be very general and intricate, many of these techniques lead to large and dense graphs. Therefore, the reduction strategy significantly affects the runtime and practicality of algorithms that run on the reduced graph. This leads to several natural questions. Are the graph sizes resulting from existing techniques inherently necessary for modeling hypergraph cuts? Given a class of functions that are known to be graph reducible, can one determine more efficient or even the \emph{most} efficient reduction techniques? Finally, is it possible to obtain more efficient reductions and faster downstream algorithms if it is sufficient to just \emph{approximately} model cut penalties?

To answer these questions, we present a novel framework for {sparsifying} hypergraph-to-graph reductions with provable guarantees on preserving cut properties. Our framework brings together concepts and techniques from several different theoretical domains, including algorithms for solving generalized hypergraph cut problems~\cite{veldt2020hypercuts,panli_submodular,panli2017inhomogeneous,yoshida2019cheeger}, standard graph sparsification techniques~\cite{soma2019spectral,spielman2014nearly,batson2014twice}, and tools for approximating functions with piecewise linear curves~\cite{magnanti2012separable,magnati2004ipco}. We present sparsification techniques for a large and natural class of submodular splitting functions that are \emph{cardinality-based}, meaning that $\vw_e(A) = \vw_e(B)$ whenever $|A| = |B|$. These are known to always be graph reducible, and are particularly natural for several downstream applications~\cite{veldt2020hypercuts}. Our approach leads to graph reductions that are significantly more sparse than previous approaches, and we show that our method is in fact optimally sparse under a certain type of reduction strategy. Our sparsification framework can be directly used to develop faster algorithms for approximately solving hypergraph $s$-$t$ cut problems~\cite{veldt2020hypercuts}, and improve runtimes for a large class of cardinality-based decomposable submodular minimization problems~\cite{kohli2009robust,kolmogorov2012minimizing,jegelka2013reflection,stobbe2010efficient}. We also show how our techniques enable us to develop efficient sparsifiers for graphs constructed from co-occurrence data.

\subsection{Graph and Hypergraph Sparsification}
Our framework and results share numerous connections with existing work on graph sparsification, which we review here. Let $G = (V,E)$ be a graph with a cut function $\cut_G$, which can be viewed as a very restricted case of the generalized hypergraph cut function in Eq.~\eqref{hcf}.
%
An $\varepsilon$-{cut sparsifier} for $G$ is a sparse weighted and undirected graph $H = (V,F)$ with cut function $\cut_H$, such that
\begin{equation}
\label{eq:cutsparse}
\cut_G(S) \leq \cut_H(S) \leq (1+\varepsilon) \cut_G(S),
\end{equation}
for every subset $S \subseteq V$. This definition was introduced by Bencz\'{u}r and Karger~\cite{benczur1996approximating}, who showed how to obtain a sparsifier with $O(n \log n/\varepsilon^2)$ edges for any graph in $O(m \log^3 n)$ time for an $n$-node, $m$-edge graph. The more general notion of \emph{spectral} sparsification, which approximately preserves the Laplacian quadratic form of a graph rather than just the cut function, was later introduced by Spielman and Teng~\cite{spielmanteng2011}. The best cut and spectral sparsifiers have $O(n/\varepsilon^2)$ edges, which is known to be asymptotically optimal for both spectral and cut sparsifiers~\cite{andoni2016sketch,batson2014twice}. Although studied much less extensively, analogous definitions of cut~\cite{checkuri2018minimum,kogan2015sketching} and spectral~\cite{soma2019spectral} sparsifiers for hypergraphs have also been developed. However, these apply exclusively to the all-or-nothing cut penalty, and do not preserve generalized cut functions of the form shown in~\eqref{hcf}. Bansal et al.~\cite{bansal2019new} also considered a weaker notion of graph and hypergraph sparsification, involving additive approximation terms, but in the present work we only consider multiplicative approximations.

\subsection{The Present Work: Augmented Sparsifiers for Hypergraph Reduction}
\label{sec:present}
In this paper, we introduce an alternative notion of an \emph{augmented} cut sparsifier. We present our results in the context of hypergraph-to-graph reductions, though our framework also provides a new notion of augmented sparsifiers for graphs.
Let $\mathcal{H} = (V,\mathcal{E})$ be a hypergraph with a generalized cut function $\cut_\mathcal{H}$, and let $\hat{G} = (V\cup \mathcal{A}, \hat{E})$ be a directed graph on an {augmented} node set $V \cup \mathcal{A}$. The graph is equipped with an {augmented} cut function defined for any $S \subseteq V$ by
\begin{equation}
\label{eq:specialcut}
\cut_{\hat{G}}(S) = \min_{T \subseteq \mathcal{A}} \,\, \dircut_{\hat{G}}(S \cup T),
\end{equation}
where $\dircut_{\hat{G}}$ is the standard \emph{directed} cut function on $\hat{G}$. 
We say that $\hat{G}$ is an $\varepsilon$-\emph{augmented cut sparsifier} for $\mathcal{H}$ if it is sparse and satisfies
\begin{equation}
\label{eq:augsparse}
\cut_\mathcal{H}(S) \leq \cut_{\hat{G}}(S) \leq (1+\varepsilon) \cut_\mathcal{H}(S).
\end{equation}
The minimization involved in~\eqref{eq:specialcut} is especially natural when the goal is to approximate a minimum cut or minimum $s$-$t$ cut in $\mathcal{H}$. If we solve the corresponding cut problem in $\hat{G}$, nodes from the \emph{auxiliary} node set $\mathcal{A}$ will be automatically arranged in a way that yields the minimum directed cut penalty, as required in~\eqref{eq:specialcut}. If $\hat{S}^*$ is the minimum cut in $\hat{G}$, $S^* = V \cap \hat{S}^*$ will be a $(1+\varepsilon)$-approximate minimum cut in $G$. 
Even when solving a minimum cut problem is not the goal, our sparsifiers will be designed in such a way that the augmented cut function~\eqref{eq:specialcut} will be easy to evaluate.

Unlike the standard graph sparsification problem, in some cases it may in fact be impossible to find \emph{any} directed graph $\hat{G}$ satisfying~\eqref{eq:augsparse}, independent of the graph's density. In recent work we showed that hypergraphs with non-submodular splitting functions are never graph reducible~\cite{veldt2020hypercuts}.  {\v{Z}}ivn{\'y} et al.~\cite{zivny2009} showed that even in the case of four-node hyperedges, there exist \emph{submodular} splitting functions (albeit \emph{asymmetric} splitting functions) that are not representable by graph cuts. Nevertheless, there are several special cases in which graph reduction is possible~\cite{kolmogorov2004what,freedman2005energy,kohli2009robust}. 

\paragraph{Augmented Sparsifiers for Cardinality-Based Hypergraph Cuts}
We specifically consider the class of submodular splitting functions that are cardinality-based, meaning they satisfy $\vw_e(A) = \vw_e(B)$ whenever $A, B \subseteq e$ satisfy $|A| = |B|$. These are known to be graph reducible~\cite{kohli2009robust,veldt2020hypercuts}, though existing techniques will reduce a hypergraph $\mathcal{H} = (V,\mathcal{E})$ to a graph with $O(|V| + \sum_{e \in \mathcal{E}} |e|)$ nodes and $O(\sum_{e \in \mathcal{E}} |e|^2)$ edges. We prove the following \emph{sparse} reduction result.
\begin{theorem}
	\label{thm:aug}
	Let $\mathcal{H} = (V,\mathcal{E})$ be a hypergraph where each $e \in \mathcal{E}$ is associated with a cardinality-based submodular splitting function. There exists an augmented cut sparsifier $\hat{G}$ for $\mathcal{H}$ with $O(|V| +  \frac{1}{\varepsilon} \sum_{e \in \mathcal{E}} \log |e|)$ nodes and $O(\frac{1}{\varepsilon}\sum_{e \in E} |e| \log |e|)$ edges.
\end{theorem}
For certain types of splitting functions (e.g., the one corresponding to a clique expansion),
we show that our reductions are even more sparse. 

\paragraph{Augmented Sparsifiers for Graphs}
Another relevant class of augmented sparsifiers to consider is the setting where $\mathcal{H}$ is simply a graph. In this case, if $\mathcal{A}$ is empty and all edges are undirected, condition~\eqref{eq:augsparse} reduces to the standard definition of a cut sparsifier. A natural question is whether there exist cases where allowing auxiliary nodes and directed edges leads to improved sparsifiers. We show that the answer is yes in the case of dense graphs constructed from co-occurrence data.

\paragraph{Augmented Spectral Sparsifiers} 
Just as spectral sparsifiers generalize cut sparsifiers in the standard graph setting, one can define an analogous notion of an augmented \emph{spectral} sparsifier for hypergraph reductions. This can be accomplished using existing hypergraph generalizations of the Laplacian operator~\cite{panli_submodular,yoshida2019cheeger,chan2018generalizing,anand2015hypergraph}. However, although developing augmented spectral sparsifiers constitutes an interesting open direction for future research, it is unclear whether the techniques we develop here can be used or adapted to spectrally approximate generalized hypergraph cut functions. We include further discussion on hypergraph Laplacians and spectral sparsifiers in Section~\ref{sec:conclusion}, and pose questions for future work. Our primary focus in this manuscript is to develop techniques for augmented \emph{cut} sparsifiers.

\subsection{Our Approach: Cut Gadgets and Piecewise Linear Functions}
\begin{figure}[t]
	\centering
		\begin{subfigure}[t]{0.33\linewidth}
		\centering
		\includegraphics[width=.95\linewidth]{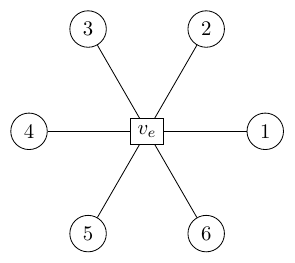}
		\caption{Star gadget}
		\label{stargad}
	\end{subfigure}\hfill
	\begin{subfigure}[t]{0.3\linewidth}
		\centering
		\includegraphics[width=.95\linewidth]{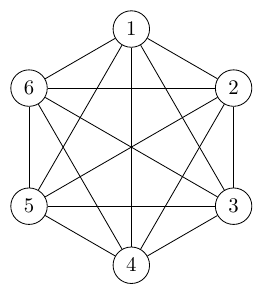}
		\caption{Clique gadget}
		\label{cliquegad}
	\end{subfigure}\hfill
	\begin{subfigure}[t]{0.33\linewidth}
		\centering
		\includegraphics[width=.95\linewidth]{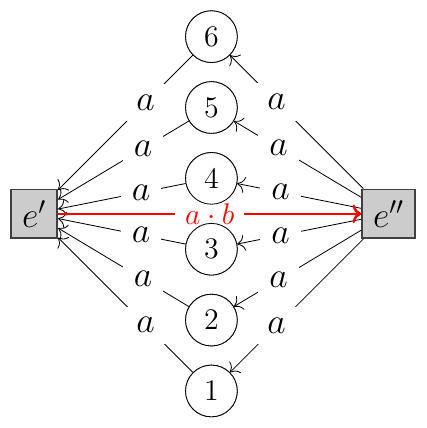}
		\caption{CB-gadget}
		\label{cbgad}
	\end{subfigure}
	\caption{Three gadgets, each modeling a different hyperedge splitting function.}
	\label{fig:gadgets}
\end{figure}
Graph reduction techniques work by replacing a hyperedge with a small graph \emph{gadget} modeling the same cut properties as the hyperedge splitting function. The simplest example of a graph reducible function is the quadratic splitting function, which we also refer to as the \emph{clique} splitting function:
\begin{equation}
\label{eq:quad}
\vw_e(S) = |A| \cdot |e\backslash A|,  \text{ for $A \subseteq e$.}
\end{equation}
This function can be modeled by replacing a hyperedge with a clique (Figure~\ref{cliquegad}). Another function that can be modeled by a gadget is the linear penalty, which can be modeled by a star gadget~\cite{zien1999}:
\begin{equation}
\label{eq:linear}
\vw_e(S) = \min \{ |A|, |e\backslash A| \}, \text{ for $A \subseteq e$.}
\end{equation}
A star gadget (Figure~\ref{stargad}) contains an auxiliary node $v_e$ for each $e \in \mathcal{E}$, which is attached to each $v \in e$ with an undirected edge. In order to model the broader class of submodular cardinality-based splitting functions, we previously introduced the cardinality-based gadget~\cite{veldt2020hypercuts} (CB-gadget) (Figure~\ref{cbgad}). This gadget is parameterized by positive scalars $a$ and $b$, and includes two auxiliary nodes $e'$ and $e''$. For each node $v \in e$, there is a directed edge from $v$ to $e'$ and a directed edge from $e''$ to $v$, both of weight $a$. Lastly, there is a directed edge from $e'$ to $e''$ of weight $a\cdot b$. This CB-gadget corresponds to the following splitting function:
\begin{equation}
\label{abcbgadget}
{\vw}_{a,b}(A) = a \cdot \min \{ |A|, |e \backslash A|, b \}.
\end{equation}
Every submodular, cardinality-based (SCB) splitting function can be modeled by a \emph{combination} of CB-gadgets with different edge weights~\cite{veldt2020hypercuts}. A different reduction strategy for minimizing submodular energy functions with cardinality-based penalties was also previously developed by Kohli et al.~\cite{kohli2009robust}. Both techniques require up to $O(k^2)$ directed edges for a $k$-node hyperedge.

\paragraph{Sparse Combinations of CB-gadgets}
\begin{figure}[t]
	\centering
	\begin{subfigure}[t]{0.32\linewidth}
		\centering
		\includegraphics[width=.9\linewidth]{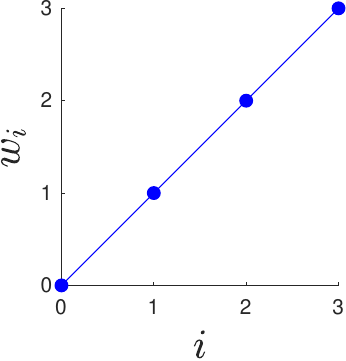}
		\caption{Star function weights}
		\label{starweights}
	\end{subfigure}
	\begin{subfigure}[t]{0.32\linewidth}
		\centering
		\includegraphics[width=.9\linewidth]{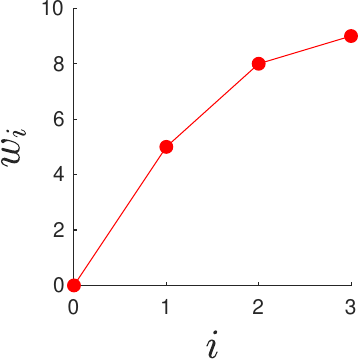}
		\caption{Clique function weights}
		\label{cliqueweights}
	\end{subfigure}
	\begin{subfigure}[t]{0.32\linewidth}
		\centering
		\includegraphics[width=.9\linewidth]{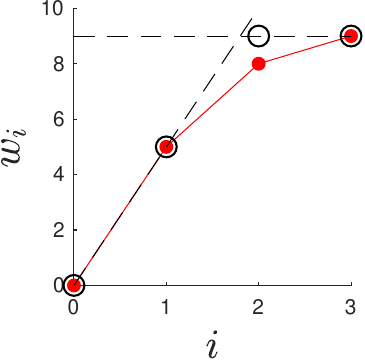}
		\caption{Piecewise linear approx}
		\label{cliqueapprox}
	\end{subfigure}
	\begin{subfigure}[t]{0.32\linewidth}
		\centering
		\includegraphics[width=.8\linewidth]{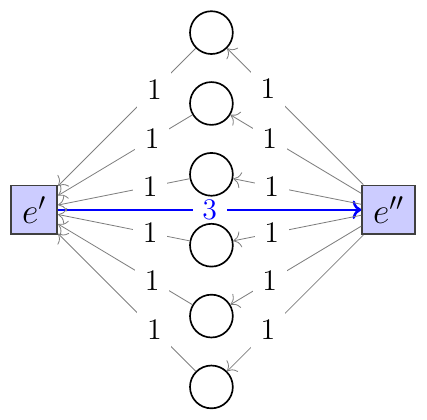}
		\caption{CB-gadget for star}
		\label{cgb_for_star}
	\end{subfigure}
	\begin{subfigure}[t]{0.32\linewidth}
		\centering
		\includegraphics[width=.7\linewidth]{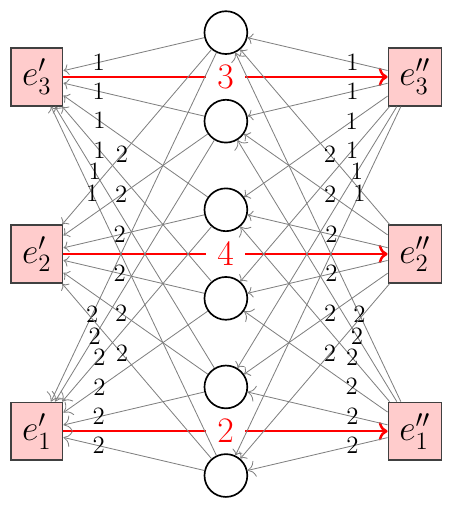}
		\caption{CB-gadgets for clique}
		\label{cgb_for_clique}
	\end{subfigure}
	\begin{subfigure}[t]{0.32\linewidth}
		\centering
		\includegraphics[width=.8\linewidth]{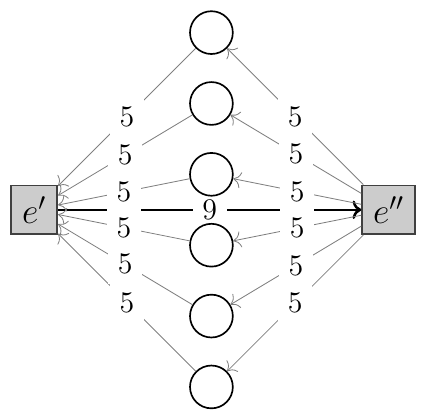}
		\caption{Approximating the clique}
		\label{sparsified_clique}
	\end{subfigure}
	\caption{(a) The linear splitting function~\eqref{eq:linear} can be modeled by a sparse gadget (d). The quadratic splitting function~\eqref{eq:quad} penalties (b) can be modeled by a dense gadget (e). A piecewise linear approximation for the quadratic splitting penalties (c) corresponds to a sparse gadget (f).}
	\label{fig:illustration}
\end{figure}
Our work introduces a new framework for \emph{approximately} modeling submodular cardinality-based (SCB) splitting functions using a small combinations of CB-gadgets.  Figure~\ref{fig:illustration} illustrates our sparsification strategy.  We first associate an SCB splitting function with a set of points $\{(i,w_i)\},$ where $i$ represents the number of nodes on the ``small side'' of a cut hyperedge, and $w_i$ is the penalty for such a split. We show that when many of these points are collinear, they can be modeled with a {smaller number} of CB-gadgets. As an example, the star expansion penalties~\eqref{eq:linear} can be modeled with a single CB-gadget (Figures~\ref{starweights} and~\ref{cgb_for_star}), whereas modeling the quadratic penalty with previous techniques~\cite{veldt2020hypercuts} requires many more (Figures~\ref{cliqueweights} and~\ref{cgb_for_clique}). Given this observation, we design new techniques for $\varepsilon$-approximating the set of points $\{(i, w_i)\}$ with a piecewise linear curve using a small number linear pieces. We then show how to translate the resulting piecewise linear curve back into a smaller combination of CB-gadgets that $\varepsilon$-approximates the original splitting function. Our piecewise linear approximation strategy allows us to find the optimal (i.e., minimum-sized) graph reduction in terms of CB-gadgets.
When $\varepsilon = 0$, our approach finds the best way to \emph{exactly} model an SCB splitting function, and requires only half the number of gadgets needed by previous techniques~\cite{veldt2020hypercuts}. More importantly, for larger $\varepsilon$, we prove the following sparse approximation result, which is used to prove Theorem~\ref{thm:aug}.
\begin{theorem}
	For $\varepsilon \geq 0$, any submodular cardinality-based splitting function on a $k$-node hyperedge can be $\varepsilon$-modeled by combining $O(\min\{\log k /\varepsilon, k\})$ CB-gadgets. 
\end{theorem}
We show that a nearly matching lower bound of $O(\log k/\sqrt{\varepsilon})$ CB-gadgets is required for modeling a square root splitting function. Despite worst case bounds, we prove that only $O(\varepsilon^{-1/2} \log \log \frac{1}{\varepsilon} )$ CB-gadgets are needed to approximate the quadratic splitting function, independent of hyperedge size. This is particularly relevant for approximating the widely used clique expansion technique, as well as for modeling certain types of dense {co-occurrence} graphs. All of our sparse reduction techniques are combinatorial, deterministic, and very simple to use in practice.

\subsection{Augmented Sparsifiers for Co-occurrence Graphs}
When $\mathcal{H}$ is just a graph, augmented sparsifiers correspond to a generalization of standard cut sparsifiers that allow directed edges and auxiliary nodes. The auxiliary nodes in this case play a role analogous to Steiner nodes in finding minimum spanning trees. Just as adding Steiner nodes makes it possible to find a smaller weight spanning tree, it is natural to ask whether including an auxiliary node set might lead to better cut sparsifiers for a graph $G$. We show that the answer is yes for certain classes of dense \emph{co-occurrence} graphs, which are graphs constructed by inducing a clique on a set of nodes that share a certain property or participate in a certain type of group interaction (equivalently, clique expansions of hypergraphs).
Steiner nodes have in fact been previously used in constructing certain types of sparsifiers called \emph{vertex} and \emph{flow} sparsifiers~\cite{chuzhoy2012}. However, these are concerned with preserving certain routing properties between distinguished terminal nodes in a graph, and are therefore distinct from our goal of obtaining $\varepsilon$-cut sparsifiers. 

\paragraph{Sparsifying the complete graph}
Our ability to sparsify the clique splitting function~\eqref{eq:quad} directly implies a new approach for sparsifying a complete graph. Cut sparsifiers for the complete graph provide a simple case study for understanding the differences in sparsification guarantees that can be obtained when we allow auxiliary nodes and directed edges. Furthermore, better sparsifiers for the complete graph can be used to design useful sparsifiers for co-occurrence graphs.
We have the following result.
\begin{theorem}
	\label{thm:complete}
	Let $G = (V,E)$ be the complete graph on $n = |V|$ nodes. There exists an $\varepsilon$-augmented sparsifier for $G$ with $O(n)$ nodes and $O(n\varepsilon^{-1/2}\log \log \frac{1}{\varepsilon})$ edges.
\end{theorem}
By comparison, the best standard cut and spectral sparsifiers for the complete graph have exactly $n$ nodes and $O(n /\varepsilon^2)$ edges. This is tight for spectral sparsifiers~\cite{batson2014twice}, as well as for degree-regular cut sparsifiers with uniform edge weights~\cite{alon1997expansion}. Thus, by adding a small number of auxiliary nodes, our sparsifiers enable us to obtain a significantly better dependence on $\varepsilon$ when cut-sparsifying a complete graph. Our sparsifier is easily constructed deterministically in $O(n\varepsilon^{-1/2}\log \log \frac{1}{\varepsilon})$ time.

Standard \emph{undirected} sparsifiers for the complete graph have received significant attention as they correspond to expander graphs~\cite{lubotzky1988ramanujan,alon1997expansion,margulis1988explicit,batson2014twice}. We remark that the directed augmented cut sparsifiers we produce are very different in nature and should not be viewed as expanders. In particular, unlike for expander graphs, random walks on our complete graph sparsifiers will converge to a very non-uniform distribution. We are interested in augmented sparsifiers for the complete graph simply for their ability to model cut properties in a different way, and the implications this has for sparsifying hypergraph clique expansions and co-occurrence graphs.

\paragraph{Sparsifying co-occurrence graphs}
Co-occurrence relationships are inherent in the construction of many types of graphs. Formally, consider a set of $n = |V|$ nodes that are organized into a set of co-occurrence interactions $\mathcal{C}  \subseteq 2^V$. Interaction $c \in \mathcal{C}$ is associated with a weight $w_c > 0$, and an edge between nodes $i$ and $j$ is created with weight $w_{ij} = \sum_{c\in \mathcal{C}: i,j \in c} w_c$. When $w_c = 1$ for every $c \in \mathcal{C}$, $w_{ij}$ equals the number of interactions that $i$ and $j$ share. We use $d_\mathit{avg}$ to denote the average number of co-occurrence interactions in which nodes in $V$ participate. The cut value in the resulting graph $G = (V,E)$ for a set $S \subseteq V$ is given by the following \emph{co-occurrence cut function}:
\begin{equation}
\label{eq:cocut}
\cut_G(S) = \sum_{c \in \mathcal{C}} w_c\cdot |S \cap c| \cdot |\bar{S} \cap c|.
\end{equation}
Graphs with this co-occurrence cut function arise frequently as clique expansions of a hypergraph~\cite{Zhou2006learning,BensonGleichLeskovec2016,zien1999,hadley1995}, or as projections of a bipartite graph~\cite{li2007evolving,neal2014backbone,newman2001random,ramasco2006social,zhou2007bipartite,destefano2013use,watts1998collective}. Even when the underlying dataset is not first explicitly modeled as a hypergraph or bipartite graph, many approaches implicitly use this approach to generate a graph from data. When enough group interaction sizes are large, $G$ becomes dense, even if $|\mathcal{C}|$ is small. We can significantly sparsify $G$ by applying an efficient sparsifier to each clique induced by a co-occurrence relationship. Importantly, we can do this without ever explicitly forming $G$. By applying Theorem~\ref{thm:complete} as a black-box for clique sparsification, we obtain the following result.
\begin{theorem}
	\label{thm:cooc}
	Let $G = (V,E)$ be the co-occurrence graph for some $\mathcal{C} \subseteq 2^V$ and let $n = |V|$. For $\varepsilon > 0$, there exists an augmented sparsifier $\hat{G}$ with $O(n + |\mathcal{C}| \cdot f(\varepsilon))$ nodes and $O(n\cdot d_\mathit{avg} \cdot f(\varepsilon) )$ edges, where $f(\varepsilon) = \varepsilon^{-1/2}\log \log \frac{1}{\varepsilon}$. In particular, if $d_{avg}$ is constant and for some $\delta > 0$ we have $\sum_{c \in \mathcal{C}} |c|^2 = \Omega(n^{1+\delta})$, then forming $G$ explicitly takes $\Omega(n^{1+\delta})$ time, but an augmented sparsifier for $G$ with $O(nf(\varepsilon))$ nodes and $O(n f(\varepsilon))$ edges can be constructed in $O( n f(\varepsilon))$ time.
\end{theorem}
Importantly, the average co-occurrence degree $d_\mathit{avg}$ is not the same as the average node degree in $G$, which will typically be much larger. Theorem~\ref{thm:cooc} highlights that in regimes where $d_\mathit{avg}$ is a constant, our augmented sparsifiers will have fewer edges than the number needed by standard $\varepsilon$-cut sparsifiers. In Section~\ref{sec:ccs}, we consider simple graph models that satisfy these assumptions. We also consider tradeoffs between our augmented sparsifiers and standard sparsification techniques for co-occurrence graphs. Independent of the black-box sparsifier we used, implicitly sparsifying $G$ in this way will often lead to significant runtime improvements over forming $G$ explicitly.

\subsection{Approximate Cardinality-based DSFM}
\begin{table}[t]
	\caption{Runtimes for Card-DFSM for various methods, where 
		$\mu = \sum_e |e|$, and $\mu_2 = \sum_e |e|^2$, and where $\theta_\mathit{max}$ and $\theta_\mathit{avg}$ denote certain oracle runtimes. These satisfy $\theta_\mathit{max} = \Omega(\max |e|)$, and $\theta_\mathit{avg} = \Omega(\frac{1}{R}\sum_e |e|)$. $T_{mf}(N,M)$ is the time to solve a max-flow problem with $N$ nodes and $M$ edges.}
	\label{tab:dfsm}
	\centering
	\begin{tabular}{l ll}
		\toprule 
		Method & Discrete/Cont & Runtime  \\
		\midrule
		Kolmogorov SF~\cite{kolmogorov2012minimizing} & Discrete & $\tilde{O}(\mu^2)$  \\
		IBFS Strong~\cite{fix2013structured,ene2017decomposable} & Discrete &$O(n^2 \theta_\mathit{max} \mu_2)$  \\
		IBFS Weak~\cite{fix2013structured,ene2017decomposable} & Discrete &$\tilde{O}(n^2 \theta_\mathit{max} + n \sum_{e} |e|^4)$ \\
		AP~\cite{nishihara2014convergence,ene2017decomposable,panli2018revisiting} & Continuous & $\tilde{O}(nR\theta_\mathit{avg}\mu)$ \\
		RCDM~\cite{ene2015random,ene2017decomposable} & Continuous & $\tilde{O}(n^2 R\theta_\mathit{avg})$ \\
		ACDM~\cite{ene2015random,ene2017decomposable} & Continuous & $\tilde{O}(n R\theta_\mathit{avg})$ \\
		Axiotis et al.~\cite{axiotis2021decomposable} & Discrete &  $\tilde{O}( \max_e |e|^2 \cdot \left(\sum_{e} |e|^2 \theta_e +  T_\mathit{mf}\left(n,n + \mu_2\right) \right)$  \\
		This paper (exact solutions) & Discrete & $\tilde{O}(T_\mathit{mf}(\mu , \mu_2 )) = \tilde{O}\left(\mu_2 + \mu^{3/2} \right)$   \\ 
		This paper (approximate solutions)& Discrete & $\tilde{O}\left( T_\mathit{mf}(n+ \frac{R}{\varepsilon}, \frac{1}{\varepsilon} \mu )\right) = \tilde{O}\left( \frac{\mu}{\varepsilon} + (n + \frac{R}{\varepsilon})^{3/2} \right)$  \\
		\bottomrule
	\end{tabular}
	\vspace{-\baselineskip}
\end{table}
Typically in hypergraph cut problems it is natural to assume that splitting functions are symmetric and satisfy $\vw_e(\emptyset) = \vw_e(e) = 0$. However, we show that our sparse reduction techniques apply even when these assumptions do not hold. This allows us to design fast algorithms for approximately solving certain decomposable submodular function minimization (DSFM) problems.
Formally a function $f \colon 2^V \rightarrow \mathbb{R}^+$ is a decomposable submodular function if it can be written as
\begin{equation}
\label{sos}
f(S) = \sum_{e \in \mathcal{E}} \vf_e(S \cap e),
\end{equation}
where each $\vf_e$ is a submodular function defined on a set $e \subseteq V$. Following our previous notation and terminology, we say $\vf_e$ is cardinality-based if $\vf_e(S) = g_e(|S|)$ for some concave function $g_e$. This special case, which we refer to as Card-DSFM, has been one of the most widely studied and applied variants since the earliest work on DSFM~\cite{kolmogorov2012minimizing,stobbe2010efficient}.
In terms of theory, previous research has addressed specialized runtimes and solution techniques~\cite{kolmogorov2012minimizing, stobbe2010efficient,kohli2009robust}.
In practice, cardinality-based decomposable submodular functions frequently arise as higher-order energy functions in computer vision~\cite{kohli2009robust} and set cover functions~\cite{stobbe2010efficient}.
Even previous research on algorithms for the more general DSFM problem tends to focus on cardinality-based examples in experimental results~\cite{ene2017decomposable,stobbe2010efficient,jegelka2013reflection,panli2018revisiting}.

Existing approaches for minimizing these functions focus largely on finding exact solutions. Using our sparse reduction techniques, we develop the first approximation algorithms for the problem. Let $n = |V|$, $R = |\mathcal{E}|$, and $\mu = \sum_{e \in \mathcal{E}} |e|$. In Appendix~\ref{app:general}, we show that a result similar to Theorem~\ref{thm:aug} also holds for more general cardinality-based splitting functions. In Section~\ref{sec:runtime}, we combine that result with fast recent $s$-$t$ cut solvers~\cite{brand2021minimum} to prove the following theorem.
\begin{theorem}
	\label{thm:card}
	Let $\varepsilon > 0$. Any cardinality-based decomposable submodular function can be minimized to within a multiplicative $(1+\varepsilon)$ factor in 
	$\tilde{O}\left( \frac{1}{\varepsilon}\sum_{e \in \mathcal{E}}|e| + (n + \frac{R}{\varepsilon})^{3/2} \right)$ time.
\end{theorem}
We compare this runtime against the best previous techniques for Card-DSFM. We summarize runtimes for competing approaches in Table~\ref{tab:dfsm}.
Our techniques enable us to highlight regimes of the problem where we can obtain significantly faster algorithms in cases where it is sufficient to solve the problem approximately. For example, whenever $n = \Omega(R)$, our algorithms for finding approximate solutions provide a runtime advantage --- often a significant one --- over approaches for computing an exact solution. We provide a more extensive theoretical runtime comparison in Section~\ref{sec:runtime}. In Section~\ref{sec:experiments}, we show that implementations of our methods lead to significantly faster results for benchmark image segmentation tasks and localized hypergraph clustering problems.


\section{The Sparse Gadget Approximation Problem}
\label{sec:spagap}
A generalized hypergraph cut function is defined as the sum of its splitting functions. Therefore, if we can design a technique for approximately modeling a single hyperedge with a sparse graph, this in turn provides a method for constructing an augmented sparsifier for the entire hypergraph. We now formalize the problem of approximating a submodular cardinality-based (SCB) splitting function using a combination of cardinality-based (CB) gadgets.
We abstract this as the task of approximating a certain class of functions with integer inputs (equivalent to SCB splitting functions), using a small number of simpler functions (equivalent to cut properties of the gadgets). Let $[r] = \{1,2, \hdots ,r\}$.
\begin{definition}
	\label{def:scb}
	An $r$-SCB integer function is a function $\vw \colon \{0\}\cup [r] \rightarrow \mathbb{R}^+$ satisfying
	\begin{align}
		\label{zeroconstraint}
		\vw(0) &= 0\\
		\label{mainconstraints}
		2\vw(j) &\geq \vw(j-1) + \vw(j+1) \text{ for $j = 1, \hdots, r-1$} \\
		\label{monotones}
		0 &\leq \vw(1) \leq \vw(2) \leq \hdots \leq \vw(r)
	\end{align}
	We denote the set of $r$-SCB integer functions by $\mathcal{S}_r$.
\end{definition}
The value $\vw(i)$ represents the splitting penalty for placing $i$ nodes on the small side of a cut hyperedge. In previous work we showed that the inequalities given in Definition~\ref{def:scb} are necessary and sufficient conditions for a cardinality-based splitting function to be submodular~\cite{veldt2020hypercuts}. The $r$-SCB integer function for a CB-gadget with edge parameters $(a,b)$ (see~\eqref{abcbgadget}) is
\begin{equation}
\label{scbcb}
\vw_{a,b}(i) = a \cdot \min \{ i, b\}\,.
\end{equation}
Combining $J$ CB-gadgets produces a \emph{combined} $r$-SCB integer function of special importance. 
\begin{definition}
	\label{def:ccb}
	An $r$-CCB  (\emph{C}ombined \emph{C}ardinality-\emph{B}ased gadget) function of order $J$, is an $r$-SCB integer function $\hat{\vw}$ with the form
	\begin{equation}
	\label{combinedscb}
	\hat{\vw}(i) = \sum_{j = 1}^J a_j \cdot \min \{ i, b_j\}\,, \text{ for $i \in [r]$}.
	\end{equation}
	where the $t$-dimensional vectors $\va = (a_j)$ and $\vb = (v_j)$ parameterizing $\hat{\vw}$ satisfy:
	\begin{align}
		\label{ab1}
		b_j &> 0, a_j > 0 \text{ for all $j \in [J]$}  \\
		\label{ab2}
		b_j &< b_{j+1} \text{ for $j \in [J-1]$}\\
		\label{ab3}
		b_J & \leq r.
	\end{align}
	We denote the set of $r$-CCB functions of order $J$ by $\mathcal{C}_r^J$.
\end{definition}
The conditions on the vectors $\va$ and $\vb$ come from natural observations about combining CB-gadgets. Condition~\eqref{ab1} ensures that we do not consider CB-gadgets where all edge weights are zero. The ordering in condition~\eqref{ab2} is for convenience; the fact that $b_j$ values are all distinct implies that we cannot collapse two distinct CB-gadgets into a single CB-gadget with new weights. For condition~\eqref{ab3}, observe that for any $b_J \geq r$, $\min \{i, b_J\} = i$ for all $i \in [r]$. For a helpful visual, note that the $r$-SCB function in~\eqref{scbcb} represents splitting penalties for the CB-gadget in Figure~\ref{cbgad}. An $r$-CCB function corresponds to a combination of CB-gadgets, as in Figures~\ref{cliqueapprox}~and~\ref{sparsified_clique}.

In previous work we showed that any combination of CB-gadgets produces a submodular and cardinality-based splitting function, which is equivalent to stating that $\mathcal{C}_r^J \subseteq \mathcal{S}_r$ for all $J \in \mathbb{N}$~\cite{veldt2020hypercuts}. Furthermore, $\mathcal{C}_r^r = \mathcal{S}_r$, since any $r$-SCB splitting function can be modeled by a combination of $r$ CB-gadgets. Our goal here is to determine how to approximate a function $\vw \in \mathcal{S}_r$ with some function $\hat{\vw} \in \mathcal{C}_r^J$ where $J \ll r$. This corresponds to modeling an SCB splitting function using a \emph{small} combination of CB-gadgets.
\begin{definition}
	\label{def:scgap}
	For a fixed $\vw \in \sr$ and an approximation tolerance parameter $\varepsilon \geq 0$, the Sparse Gadget Approximation Problem (\spagap) is the following optimization problem:
	\begin{equation}
	\label{eq:mainobj}
	\begin{array}{ll}
	\minimize\,\, &\kappa \\
	\text{subject to } & \vw  \leq  \hat{\vw} \leq (1+ \varepsilon) \vw\\
	& \hat{\vw} \in \mathcal{C}_r^\kappa.
	\end{array}
	\end{equation}
\end{definition}

\paragraph{Upper Bounding Approximations}
Problem~\eqref{eq:mainobj} specifically optimizes over functions $\hat{\vw}$ that upper bound $\vw$. This restriction simplifies several aspects of our analysis without any practical consequence. For example, we could instead fix some $\delta \geq 1$ and optimize over functions $\tilde{\vw}$ satisfying $\frac{1}{\delta} \vw  \leq  \tilde{\vw} \leq \delta \vw$.
However, this implies that the function $\hat{\vw} = \delta \tilde{\vw}$ satisfies
$\vw  \leq  \hat{\vw} \leq (1+ \varepsilon) \vw$,
with $\varepsilon = \delta^2 - 1$. Thus, the problems are equivalent for the correct choice of $\delta$ and $\varepsilon$.

\paragraph{Motivation for Optimizing over CB-gadgets}
A natural question to ask is whether it would be better to search for a sparsest approximating gadget over a broader classes of gadgets. There are several key reasons why we restrict to combinations of  CB-gadgets. First of all, we already know these can model \emph{any} SCB splitting function, and thus they provide a very simple building block with broad modeling capabilities. Furthermore, it is clear how to define an optimally sparse combination of CB-gadgets: since all CB-gadgets for a $k$-node hyperedge have the same number of auxiliary nodes and directed edges, an optimally sparse reduction is one with a minimum number of CB-gadgets. If we instead wish to optimize over all possible gadgets, it is likely that the \emph{best} reduction technique will depend on the splitting function that we wish to approximate. Furthermore, the optimality of a gadget may not even be well-defined, since one must take into account both the number of auxiliary nodes as well as the number of edges that are introduced, and the tradeoff between the two is not always clear. Finally, as we shall see in the next section, by restricting to CB-gadgets, we are able to draw a useful connection between sparse gadgets and approximating piecewise linear curves with a smaller number of linear pieces. 

\section{Sparsification via Piecewise Linear Approximation}
\label{sec:reductions}
We begin by defining the class of piecewise linear functions in which we are interested.
\begin{definition}
	\label{def:plfun}
	For $r \in \mathbb{N}$, $\mathcal{F}_r$ is the class of functions $\vf \colon [0,\infty] \longrightarrow \mathbb{R}_+$ such that:
	\begin{enumerate}[itemsep=0mm]
		\item $\vf(0) = 0$ 
		\item $\vf$ is a constant for all $x \geq r$
		\item $\vf$ is increasing: $x_1 \leq x_2 \implies \vf(x_1) \leq \vf(x_2)$ 
		\item $\vf$ is piecewise linear
		\item $\vf$ is concave (and hence, continuous).
	\end{enumerate}
\end{definition}
It will be key to keep track of the number of linear pieces that make up a given function $\vf \in \mathcal{F}_r$. Let $\mathcal{L}$ be the set of linear functions with nonnegative slopes and intercept terms:
\begin{equation}
\mathcal{L} = \{ g(x) = mx + d \, |\, m, d \in \mathbb{R}^+\}.
\end{equation}
Every function $\vf \in \mathcal{F}_r$ can be characterized as the lower envelope of a set of these linear functions. 
\begin{align}
	\label{Lfun}
	\vf(x) &= \min_{g \in L} g(x), \text{ where }  L \subset \mathcal{L}.
\end{align}
We use $|L|$ to denote the number of linear pieces of $\vf$. In order for~\eqref{Lfun} to properly characterize a function in $\mathcal{F}_r$, it must be constant for all $x \geq r$ (property 2 in Definition~\ref{def:plfun}), and thus $L$ must contain exactly one line of slope zero. 
The \emph{continuous extension} $\hat{\vf}$ of an $r$-CCB function $\vw$ parameterized by $(\va,\vb)$ is defined as 
\begin{equation}
\label{cextt}
\hat{\vf}(x) = \sum_{j = 1}^J a_j \cdot \min \{ x, b_j \} \text{ for $x\in [0, \infty]$}.
\end{equation}
We prove that continuously extending any $r$-CCB function always produces a function in $\mathcal{F}_r$. Conversely, every $\vf \in \mathcal{F}_r$ is the continuous extension of some $r$-CCB function. Appendix~\ref{app:proofs} provides proofs for these results.
\begin{lemma}
	\label{lem:ccb2pl}
	Let $\hat{\vf}$ be the continuous extension for $\vw$, shown in~\eqref{cextt}. This function is in the class $\mathcal{F}_r$, and has exactly $J$ positive sloped linear pieces, and one linear piece of slope zero.
\end{lemma}
\begin{lemma}
	\label{lem:pl2ccb}
	Let $\vf$ be a function in $\mathcal{F}_r$ with $J+1$ linear pieces. Let $b_i$ denote the $i$th breakpoint of $\vf$, and $m_i$ denote the slope of the $i$th linear piece of $\vf$. Define vectors $\va, \vb \in \mathbb{R}^J$ where $\vb(i) = b_i$ and $\va(i) = a_i = m_i - m_{i+1}$ for $i \in [J]$. If ${\vw}$ is the $r$-CCB function parameterized by vectors $(\va, \vb)$, then $\vf$ is the continuous extension of ${\vw}$.
\end{lemma} 

\subsection{The Piecewise Linear Approximation Problem}
Let $\vw \in \mathcal{S}_r$ be an arbitrary SCB integer function.
Lemma~\ref{lem:pl2ccb} implies that if we can find a piecewise linear function $\vf$ that approximates $\vw$ and has few linear pieces, we can extract from it a CCB function $\hat{\vw}$ with a small order $J$ that approximates $\vw$. Equivalently, we can find a sparse gadget that approximates an SCB splitting function of interest. Our updated goal is therefore to solve the following piecewise linear approximation problem, for a given $\vw \in \sr$ and $\varepsilon \geq 0$:
\begin{equation}
\label{eq:plobj}
\begin{array}{ll}
\minimize_{L \subset \mathcal{L}}\,\, & |L|  \\
\text{subject to } & \vw(i)  \leq  \vf(i) \leq (1+ \varepsilon) \vw(i) \text{ for $i \in [r]$}\\
& \vf \in \mathcal{F}_r\\
&\vf(x) =  \min_{g \in L} g(x)\\
& \text{for each $g \in L$, } g(j) = \vw(j)  \text{ for some $j \in \{0\}\cup [r]$}.
\end{array}
\end{equation}
The last constraint ensures that each linear piece $g \in L$ we consider crosses through at least one point $(j, \vw(j))$. We can add this constraint without loss of generality; if any linear piece $g$ is strictly greater than $\vw$ at all integers, we could obtain an improved approximation by scaling $g$ until it is tangent to $\vw$ at some point. This constraint, together with the requirement $\vf \in \mathcal{F}_r$, implies that the constant function $g^{(r)}(x) = \vw(r)$ is contained in every set of linear functions $L$ that is feasible for~\eqref{eq:plobj}. Since all feasible solutions contain this constant linear piece, our focus is on determining the optimal set of positive-sloped linear pieces needed to approximate $\vw$. 

\textbf{Optimal linear covers.} Given a fixed $\varepsilon \geq 0$ and $i \in  \{0\}\cup [r-1]$, we will say a set $L \subset \mathcal{L}$ is a \emph{linear cover} for a function $\vw \in \sr$ over the range $R = \{i, i+1, \hdots , r\}$, if each $g \in L$ upper bounds $\vw$ at all points, and if for each $j \in R$ there exists $g \in L$ such that $g(j) \leq (1+\varepsilon) \vw(j)$. The set $L$ is an \emph{optimal} linear cover if it contains the minimum number of positive-sloped linear pieces needed to cover $R$.
Thus, an equivalent way of expressing~\eqref{eq:plobj} is that we wish to find an optimal linear cover for $\vw$ over the interval $\{0\}\cup [r]$. In practice there may be many different function $\vf \in \mathcal{F}_r$ which solve~\eqref{eq:plobj}, but for our purposes it suffices to find one.

\subsection{Properties of Linear Pieces in the Cover}
We solve problem~\eqref{eq:plobj} by iteratively growing a set of linear functions $L \subset \mathcal{L}$ one function at a time, until all of $\vw$ is covered. Let $\vf$ be the piecewise linear function we construct from linear pieces in $L$.
\begin{figure}[t]
	\centering
	\centering
	\includegraphics[width=.5\linewidth]{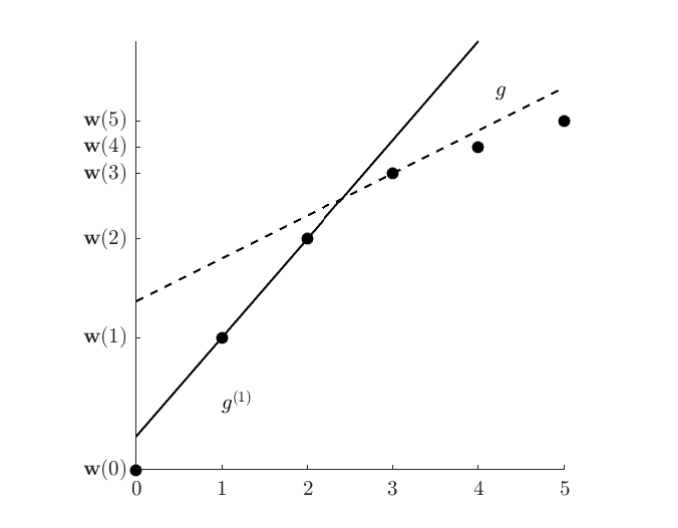}
	\caption{We restrict our attention to lines in $\mathcal{L}$ that coincide with $\vw$ at at least one integer value. Thus, every function we consider is incident to two consecutive values of $\vw$ (e.g., the solid line, $g^{(1)}$), or, it touches $\vw$ at exactly one point (dashed line, $g$).}
	\label{fig:twotypes}
\end{figure}
In order for $\vf$ to upper bound $\vw$, every function $g \in L$ in problem~\eqref{eq:plobj} must upper bound $\vw$ at every $i \in \{0\}\cup [r]$. One way to obtain such a linear function is to connect two consecutive points of $\vw$. For $i \in \{0\}\cup [r-1]$, we denote the line joining points $(i, \vw(i))$ and $(i+1, \vw(i+1))$ by
\begin{equation}
\label{gi}
g^{(i)}(x) = M_i (x - i) + \vw(i),
\end{equation}
where the slope of the line is $M_i = \vw(i+1) - \vw(i)$. 
%
In order for a line to upper bound $\vw$ but only pass through a \emph{single} point $(i, \vw(i))$ for some $i \in [r-1]$, it must have the form
\begin{equation}
\label{onept}
g(x) = m (x - i) + \vw(i) \,,
\end{equation}
where the slope $m$ satisfies  $M_{i} < m < M_{i-1}$. The existence of such a line $g$ is only possible when the points $(i-1, \vw(i-1))$, $(i, \vw(i))$, and $(i+1, \vw(i+1))$ are not collinear. To understand the strict bounds on $m$, note that if $g$ passes through $(i, \vw(i))$ and has slope exactly $M_{i-1}$, then $g$ is in fact the line $g^{(i-1)}$ and also passes through $(i-1, \vw(i-1))$. If $g$ has slope greater than $M_{i-1}$, then $g(i-1) < \vw(i-1)$ and does not upper bound $\vw$ everywhere. We can similarly argue that the slope of $g$ must be strictly greater than $M_i$ so that it does not touch or cross below the point $(i+1, \vw(i+1))$. 

We illustrate both types of functions~\eqref{gi} and~\eqref{onept} in Figure~\ref{fig:twotypes}. The following simple observation will later help in comparing approximation properties of different functions in $\mathcal{L}$.
\begin{observation}
	\label{obs:gh}
	For a fixed $\vw \in \sr$, let $g,h \in \mathcal{L}$ both upper bound $\vw$ at all integers $i \in \{0\}\cup [r]$, and assume that for some $j \in \{0\}\cup [r]$, $g(j) = h(j) = \vw(j)$. If $m_g$ and $m_h$ are the slopes of $g$ and $h$ respectively, and $m_g \geq m_h \geq 0$, then 
	\begin{itemize}
		\item For every integer $i \in [0,j]$, $\vw(i) \leq g(i) \leq h(i)$ 
		\item For every integer $i \in [j,r]$, $\vw(i) \leq h(i) \leq g(i)$.
	\end{itemize}
\end{observation}
In other words, if $g$ and $h$ are both tangent to $\vw$ at the same point $j$, but $g$ has a larger slope than $h$, then $g$ provides a better approximation for values smaller than $j$, while $h$ is the better approximation for values larger than $j$.


\subsection{Building an Optimal Linear Cover}
\label{sec:findbest}
\paragraph{The first linear piece.}
Every set $L$ solving~\eqref{eq:plobj} must include a linear piece that goes through the origin, so that $\vf(0) = 0$. We specifically choose $g^{(0)}(x) = (\vw(1) - \vw(0))x + \vw(0) = \vw(1) x$ to be the first linear piece in the set $L$ we construct. Given this first linear piece, we can then compute the largest integer $i \in [r]$ for which $g^{(0)}$ provides a $(1+\varepsilon)$-approximation:
\[
p= \max\, \{i \in [r] \;\vert\;  g^{(0)}(i) \leq (1+\varepsilon) \vw(i)\}.
\]
The integer $\ell = p+1$ therefore is the smallest integer for which we \emph{do not} have a $(1+\varepsilon)$-approximation. If $\ell \leq r$, our task is then to find the smallest number of additional linear pieces in order to cover $\{\ell, \hdots, r\}$ with $(1+\varepsilon)$-approximations. By Observation~\ref{obs:gh}, any other $g \in \mathcal{L}$ with $g(0) = 0$ and $g(1) > \vw(1)$ will be a worse approximation to $\vw$ at all integer values: $\vw(i) \leq g^{(0)}(i) < g(i)$ for all $i \in [r]$. Therefore, as long as we can find a minimum set of additional linear pieces which provides a $(1+\varepsilon)$-approximation for all $\{\ell, \hdots, r\}$, our set of functions $L$ will optimally solve objective~\eqref{eq:plobj}.



\paragraph{Iteratively finding the next linear piece.}
Consider now a generic setting in which we are given a left integer endpoint $\ell$ and we wish to find linear pieces to approximate the function $\vw$ from $\ell$ to $r$. 
We first check whether the constant function $g^{(r)}(x) = \vw(r)$ provides the desired approximation:
\begin{equation}
\label{checkr}
g^{(r)}(\ell) \leq (1+\varepsilon) \vw(\ell).
\end{equation}
If so, we augment $L$ to include $g^{(r)}$ and we are done, since this implies that $g^{(r)}$ also provides at least a $(1+\varepsilon)$-approximation at every $i \in \{\ell, \ell+1, \hdots , r\}$. If~\eqref{checkr} is not true, we must add another positive-sloped linear function to $L$ in order to get the desired approximation for all $i \in [r]$. 
We adopt a greedy approach that chooses the next line to be the optimizer of the following objective
\begin{equation}
\label{eq:local}
\begin{array}{ll}
\max_{g \in \mathcal{L}}\,\, & p'  \\
\text{subject to } & \vw(j) \leq g(j) \leq (1+\varepsilon) \vw(j) \text{ for } j = \ell, \ell + 1, \ldots, p'. \\
\end{array}
\end{equation}
In other words, solving problem~\eqref{eq:local} means finding a function that provides at least a $(1+\varepsilon)$-approximation from $\ell$ to as far towards $r$ as possible in order to cover the widest possible contiguous interval with the same approximation guarantee.
(There is always a feasible point by adding a line $g$ tangent to $\vw(\ell)$.)
The following Lemma will help us prove that this greedy scheme produces an optimal cover for $\vw$.
\begin{lemma}
	\label{lem:local}
	Let $p^*$ the solution to~\eqref{eq:local} and $g^*$ be the function that achieves it. If $\hat{L} \subset \mathcal{L}$ is an optimal cover for $\vw$ over the integer range $\{p^*+1, p^*+2, \hdots r\}$, then $\{g^*\}\cup \hat{L}$ is an optimal cover for $\{\ell, \ell + 1, \hdots r\}$.
\end{lemma}
\begin{proof}
	Let $\tilde{L}$ be an arbitrary optimal linear cover for $\vw$ over the range $\{\ell, \ell+1, \hdots, r\}$. This means that $|\hat{L} \cup \{g^*\}| \geq |\tilde{L}|$. We know $\tilde{L}$ must contain a function $g$ such that $g(\ell) \leq (1+\varepsilon) \vw(\ell)$. Let $p_g$ be the largest integer satisfying $g(p_g) \leq (1+\varepsilon) \vw(p_g)$. By the optimality of $p^*$ and $g^*$, we know $p^* \geq p_g$. Therefore, the set of functions $\tilde{L} - \{g\}$ must be a cover for the set $\{p_g +1, p_g +2, \hdots r \} \supseteq \{p^*+1, p^*+2, \hdots r\}$. Since $\hat{L}$ is an optimal cover for a subset of the integers covered by $\tilde{L} - \{g\}$, 
	\[|\hat{L}| \leq |\tilde{L} - \{g\} | \implies |\hat{L}| +1 \leq |\tilde{L}| \implies |\hat{L} \cup \{g^*\}| \leq |\tilde{L}|.\]
	Therefore, $|\hat{L} \cup \{g^*\}| = |\tilde{L}|$, so the result follows.
\end{proof}
We illustrate a simple procedure for solving~\eqref{eq:local} in Figure~\ref{fig:findnext}. The function $g$ solving~\eqref{eq:local} must either join two consecutive points of $\vw$ (the form given in~\eqref{gi}), or coincide at exactly one point of $\vw$ (form given in~\eqref{onept}).
%
%
We first identify the integer $j^*$ such that
\begin{align*}
	g^{(j^*)}(\ell) &\leq (1+\varepsilon)\vw(\ell) \\
	g^{(j^*+1)}(\ell) &> (1+\varepsilon)\vw(\ell).
\end{align*}
In other words, the linear piece connecting $(j^*, \vw(j^*))$ and $(j^*+1, \vw(j^*+1))$ provides the needed approximation at the left endpoint $\ell$, but $g^{(i)}$ for every $i > j^*$ does not. Therefore, the solution to~\eqref{eq:local} has a slope $m \in [M_{j^*}, M_{j^*+1})$, and passes through the point $(j^*, \vw(j^*))$. By Observation~\ref{obs:gh}, the line passing through this point with the smallest slope is guaranteed to provide the best approximation for all integers $p \geq j^*$. To minimize the slope of the line while still preserving the needed approximation at $\vw(\ell)$, we select the line passing through the points $(\ell, (1+\varepsilon) \vw(\ell))$ and $(j^*, \vw(j^*))$. This is given by
\begin{equation}
\label{eq:nextfun}
g^*(x) = \frac{\vw(j^*) - (1+\varepsilon)\vw(\ell)}{(j^* - \ell)}(x - \ell) + (1+\varepsilon)\vw(\ell).
\end{equation}
After adding this function $g^*$ to $L$, we find the largest integer $p \leq r$ such that $g^*(p) \leq (1+\varepsilon) \vw(p)$. If $p < r$, then we still need to find more linear pieces to approximate $\vw$, so we continue with another iteration. If $p = r$ exactly, then we do not need any more positive-sloped linear pieces to approximate $\vw$. However, we still add the constant function $g^{(r)}$ to $L$ before terminating. This guarantees that the function $\vf(x) = \min_{g \in L}(x)$ we return is in fact in $\mathcal{F}_r$. Furthermore, adding the constant function serves to improve the approximation, without affecting the order of the CCB function we will obtain from $\vf$ by applying Lemma~\ref{lem:pl2ccb}.
\begin{figure}[t]
	\centering
	\includegraphics[width=.6\linewidth]{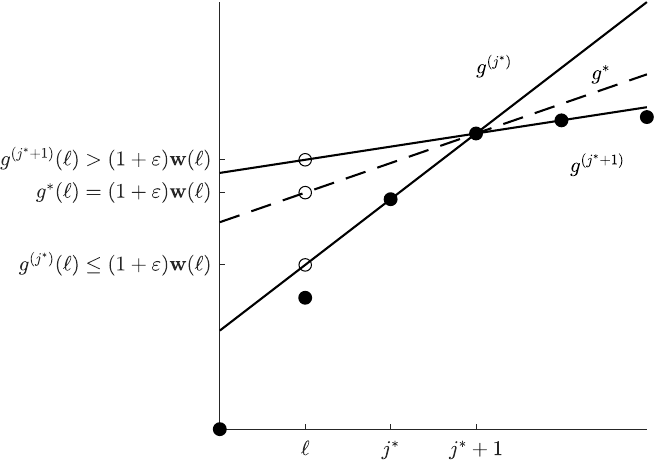}
	\caption{Given a left endpoint $\ell$ for which we do not yet have a $(1+\varepsilon)$-approximate piece, we find the next linear piece by choosing a function $g^*$ that provides the desired approximation at $\ell$, while also providing a good approximation for as large of an integer $p > \ell$ as possible. 
	}
	\label{fig:findnext}
\end{figure}

\begin{algorithm}[t]
	\caption{$\textsc{FindBest-PL-Approx}(\vw, \varepsilon)$ (solves~\eqref{eq:plobj})}
	\label{alg:findbest}
	\begin{algorithmic}
		\State \textbf{Input:} $\vw \in \sr$, $\varepsilon \geq 0$
		\State \textbf{Output:} $\vf \in \mathcal{F}_r$ optimizing~\eqref{eq:plobj}
		\State $L = \{g^{(0)}\}$, where $g^{(0)} = \vw(1)x$
		\State $p = \max\, \{i \in [r] \;\vert\; g^{(0)}(i) \leq (1+\varepsilon)\vw(i)\}$
		\State $\ell = p+1$
		\While{$\ell \leq r$}
		\State $(g^*,p) = \textsc{FindNext}(\vw, \varepsilon, \ell)$ 
		\State $\ell \leftarrow p +1$
		\State $L \leftarrow L \cup \{g^*\}$
		\If{$p = r$} 
		\State $L \leftarrow L \cup \{g^{(r)}\}$, where $g^{(r)}(x) = \vw(r)$
		\EndIf
		\EndWhile
		\State Return $\vf$ defined by $\vf(x) = \min_{g \in L} g(x)$
	\end{algorithmic}
\end{algorithm}
\begin{algorithm}[t]
	\caption{$\textsc{FindNext}(\vw, \varepsilon, \ell)$ (solves~\eqref{eq:local})}
	\label{alg:findnext}
	\begin{algorithmic}
		\State \textbf{Input:} $\vw \in \sr$, $\varepsilon \geq 0$, $\ell \in [r]$
		\State \textbf{Output:} $g \in \mathcal{L}$ optimizing~\eqref{eq:local}
		\If{$\vw(r) \leq (1+\varepsilon) \vw(\ell)$}
		\State Return $(g^{(r)},r+1)$, where $g^{(r)}(x) = \vw(r)$
		\Else
		\State $j^* = \ell$
		\While{$g^{(j^*+1)}(\ell) \leq (1+\varepsilon) \vw(\ell)$}
		\State $j^* = j^* + 1$
		\EndWhile
		\State $g^*(x) = \frac{\vw(j^*) - (1+\varepsilon)\vw(\ell)}{(j^* - \ell)}(x - \ell) + (1+\varepsilon)\vw(\ell)$
		\State $p = \max\, \{ i \in [r] \;\vert\; g^*(p) \leq (1+\varepsilon) \vw(p)$\}
		\State Return $(g^*, p)$
		\EndIf
	\end{algorithmic}
\end{algorithm}
Pseudocode for our procedure for constructing a set of function $L$ is given in Algorithm~\ref{alg:findbest}, which relies on Algorithm~\ref{alg:findnext} for solving~\eqref{eq:local}. We summarize with a theorem about the optimality of our method for solving~\eqref{eq:plobj}. 
\begin{theorem}
	Algorithm~\ref{alg:findbest} runs in $O(r)$ time and returns a function $\vf$ that optimizes~\eqref{eq:plobj}.
\end{theorem}
\begin{proof}
	The optimality of the algorithm follows by inductively applying Lemma~\ref{lem:local} at each iteration of the algorithm. For the runtime guarantee, note first of all that we can compute and store all slopes and intercepts for linear pieces $g^{(i)}$ (as given in~\eqref{gi}) in $O(r)$ time and space. As the algorithm progresses, we visit each integer $i \in [r]$ once, either to perform a comparison of the form $g^{(i)}(\ell) \leq (1+ \varepsilon) \vw(\ell)$ for some left endpoint $\ell$, or to check whether $g^*(i) \leq (1+\varepsilon) \vw(i)$ for some linear piece $g^*$ we added to our linear cover $L$. Each such $g^*$ can be computed in constant time, and as a loose bound we know we compute at most $O(r)$ such linear pieces for any $\varepsilon$.
\end{proof}

By combining Algorithm~\ref{alg:findnext} and Lemma~\ref{lem:pl2ccb}, we are able to efficiently solve \spagap.
\begin{theorem}
	Let $\vf$ be the solution to~\eqref{eq:plobj}, and $\hat{\vw}$ be the CCB function obtained from Lemma~\ref{lem:pl2ccb} based on $\vf$. Then $\hat{\vw}$ optimally solves the sparse gadget approximation problem~\eqref{eq:mainobj}.
\end{theorem}
\begin{proof}
	Since $\vf$ and $\hat{\vw}$ coincide at integer values, and $\vf$ approximates $\vw$ at integer values, we know $\vw(i)  \leq  \hat{\vw}(i) \leq (1+ \varepsilon) \vw(i)$ for $i \in [r]$. Thus, $\hat{\vw}$ is feasible for objective~\eqref{eq:mainobj}. If $\kappa^*$ is the number of positive-sloped linear pieces of $\vf$, then the order of $\hat{\vw}$ is $\kappa^*$ by Lemma~\ref{lem:pl2ccb}, and this must be optimal for~\eqref{eq:mainobj}. If it were not optimal, this would imply that there exists some upper bounding CCB function $\vw'$ of order $\kappa' < \kappa^*$ that approximates $\vw$ to within $1+\varepsilon$. But by Lemma~\ref{lem:ccb2pl}, this would imply that the continuous extension of $\vw'$ is some $\vf' \in \mathcal{F}_r$ with exactly $\kappa'$ positive-sloped linear pieces that is feasible for objective~\eqref{eq:plobj}, contradicting the optimality of $\vf$.
\end{proof}


\section{Bounding the Size of the Optimal Reduction}
\label{sec:bounds}
In our last section we showed an efficient strategy for finding the minimum number of linear pieces needed to approximate an SCB integer function. We now consider bounds on the number of needed linear pieces in different cases, and highlight implications for sparsifying hyperedges with SCB splitting functions. In the worst case, we show that we need $O(\log k / \varepsilon)$ gadgets, where $k$ is the size of the hyperedge. Moreover, this is nearly tight for the square root splitting function. 
Finally, we show that we only need $O(\varepsilon^{-1/2} \log \log \frac{1}{\varepsilon})$ gadgets to approximate the clique splitting function. This result is useful for sparsifying co-occurrence graphs and clique expansions of hypergraphs.

\subsection{The $O(\log k/\varepsilon)$ Upper Bound}
We begin by showing that a logarithmic number of CB-gadgets is sufficient to approximate any SCB splitting function.

\begin{theorem}
	\label{thm:logk}
	Let $\varepsilon \geq 0$ and $\vw_e$ be an SCB splitting function on a $k$-node hyperedge. There exists a set of $O(\log_{1+\varepsilon} k)$ CB-gadgets, which can be constructed in $O(k \log_{1+\varepsilon} k)$ time, whose splitting function $\hat{\vw}_e$ satisfies $\vw_e(A) \leq \hat{\vw}_e(A) \leq (1+\varepsilon) \vw_e(A)$ for all $A \subseteq e$.
\end{theorem}
\begin{proof}
	Let $r = \floor*{k/2}$, and let $\vw \in \mathcal{S}_r$ be the SCB integer function corresponding to $\vw_e$, i.e., $\vw(i) = \vw_e(A)$ for $A \subseteq e$ such that $|A| \in \{ i, k - i\}$. If we join all points of the form $(i, \vw(i))$ for $i \in [r]$ by a line, this results in a piecewise linear function $\vf \in \mathcal{F}_r$ that is concave and increasing on the interval $[0, r]$. We first show that there exists a set of $O(\log_{1+\varepsilon} r)$ linear pieces that approximates $\vf$ on the entire interval $[1,r]$ to within a factor $(1+\varepsilon)$. Our argument follows similar previous results for approximating a concave function with a logarithmic number of linear pieces~\cite{magnanti2012separable,gan2020graph}. 
	
	For any value $y \in [1,r]$, not necessarily an integer, $\vf(y)$ lies on a linear piece of $\vf$ which we will denote by $g^{(y)}(x) = M_y \cdot x + B_y$, where $M_y \geq 0$ is the slope and $B_y \geq 0$ is the intercept. When $y = i$ is an integer, it may be the breakpoint between two distinct linear pieces, in which case we use the rightmost line so that $g^{(y)} = g^{(i)}$ as in~\eqref{gi}, so $g^{(i)}(x) = M_i \cdot x+ B_i$ where $M_i = \vw(i+1) - \vw(i)$ and $B_i = \vw(i) - M_i \cdot i$. For any $z \in (y, r)$, the line $g^{(y)}$ provides a $z/y$ approximation to $\vf(z) = g^{(z)}(z)$, since
	\begin{align*}
	g^{(y)}(z) = M_y \cdot z + B_y \leq \frac{z}{y}( M_y \cdot y + B_y) = \frac{z}{y}\vf(y) \leq \frac{z}{y}\vf(z).
	\end{align*}
	Equivalently, the line $g^{(y)}$ provides a $(1+\varepsilon)$-approximation for every $z \in [y, (1+\varepsilon)y]$. Thus, it takes $J$ linear pieces to cover the set of intervals $[1, (1+\varepsilon)], [(1+\varepsilon), (1+\varepsilon)^2], \hdots, [(1+\varepsilon)^{J-1}, (1+\varepsilon)^{J}]$ for a positive integer $J$, and overall at most $1 +\lceil \log_{1+\varepsilon} r \rceil$ linear pieces to cover all of $[0,r]$.
	
	Since Algorithm~\ref{alg:findbest} finds the \emph{smallest} set of linear pieces to $(1+\varepsilon)$-cover the splitting penalties, this smallest set must also have at most $O(\log_{1+\varepsilon} r)$ linear pieces. Given this piecewise linear approximation, we can use Lemma~\ref{lem:pl2ccb} to extract a CCB function $\hat{\vw}$ of order $J = O(\log_{1+\varepsilon} r)$ satisfying $\vw(i) \leq \hat{\vw}(i) \leq (1+\varepsilon) \vw(i)$ for $i \in \{0\}\cup [r]$. This $\hat{\vw}$ in turn corresponds to a set of $J$ CB-gadgets that $(1+\varepsilon)$-approximates the splitting function $\vw_e$. Computing edge weights for the CB-gadgets using Algorithm~\ref{alg:findbest} and Lemma~\ref{lem:pl2ccb} takes only $O(r)$ time, so the total runtime for constructing the combined gadgets is equal to the number of individual edges that must be placed, which is $O(k \log_{1+\varepsilon} k)$. 
\end{proof}
Theorem~\ref{thm:aug} on augmented sparsifiers follows as a corollary of Theorem~\ref{thm:logk}. 
Given a hypergraph $\mathcal{H} = (V,\mathcal{E})$ where each hyperedge has an SCB splitting function, we can use Theorem~\ref{thm:logk} to expand each  $e \in \mathcal{E}$ into a gadget that has $O(\log_{1+\varepsilon} |e|)$ auxiliary nodes and $O(|e| \log_{1+\varepsilon} |e|)$ edges. Since $\log_{1+\varepsilon} n$ behaves as $\frac{1}{\varepsilon} \log n$ as $\varepsilon \rightarrow 0$, Theorem~\ref{thm:aug} follows. 

In Appendix~\ref{app:general}, we show that using a slightly different reduction, we can prove that Theorem~\ref{thm:logk} holds even when we do not require splitting functions to be symmetric or satisfy $\vw_e(\emptyset) = \vw_e(e) = 0$. In Section~\ref{sec:runtime} we use this fact to develop approximation algorithms for cardinality-based decomposable submodular function minimization.

\subsection{Near Tightness on the Square Root Function}
Next we show that our upper bound is nearly tight for the square root $r$-SCB integer function,
\begin{equation}
\label{eq:sqrt}
\vw(i) = \sqrt{i}  \text{ for $i \in \{0\}\cup [r]$}.
\end{equation}
For this result, we rely on a result previously shown by Magnanti and Stratila~\cite{magnanti2012separable} on the number of linear pieces needed to approximate the square root function over a {continuous} interval. 
\begin{lemma}
	\label{thm:magnanti}
(Lemma 3 in~\cite{magnanti2012separable}) Let $\varepsilon > 0$ and $\phi(x) = \sqrt{x}$. Let $\psi$ be a piecewise linear function whose linear pieces are all tangent lines to $\phi$, satisfying
$\psi(x) \leq (1+\varepsilon) \phi(x)$ for all $x \in [l, u]$ for $0 < l < u$. Then $\psi$ contains at least $\lceil \log_{\gamma(\varepsilon) }\frac{u}{l} \rceil$ linear pieces, where $\gamma(\varepsilon) = (1+ 2\varepsilon(2+\varepsilon) + 2(1+\varepsilon)\sqrt{\varepsilon(2+\varepsilon)})^2$. 
There exists a piecewise linear function $\psi^*$ of this form with exactly $\lceil \log_{\gamma(\varepsilon) }\frac{u}{l} \rceil$ linear pieces.\footnote{This additional statement is not included explicitly in the statement of Lemma 3 in~\cite{magnanti2012separable}, but it follows directly from the proof of the lemma, which shows how to construct such an optimal function $\psi^*$.} As $\varepsilon \rightarrow 0$, this values behaves as $\frac{1}{\sqrt{32\varepsilon}} \log \frac{u}{l}$. 
\end{lemma}
Lemma~\ref{thm:magnanti} is concerned with approximating the square root function for \emph{all} values on a \emph{continuous} interval. Therefore, it does not immediately imply any bounds on approximating a discrete set of splitting penalties. In fact, we know that when lower bounding the number of linear pieces needed to approximate any $\vw \in \mathcal{S}_r$, there is no lower bound of the form $q(\varepsilon) f(r)$ that holds for \emph{all} $\varepsilon > 0$, if $q$ is a function such that $q(\varepsilon) \rightarrow \infty$ as $\varepsilon \rightarrow 0$. This is simply because we can approximate $\vw$ by piecewise linear interpolation, leading to an upper bound of $O(r)$ linear pieces even when $\varepsilon  = 0$. Therefore, the best we can expect is a lower bound that holds for $\varepsilon$ values that may still go to zero as $r \rightarrow \infty$, but are bounded in such a way that we do not contradict the $O(r)$ upper bound that holds for all SCB integer functions. We prove such a result for the square root splitting function, using Lemma~\ref{thm:magnanti} as a black box. When $\varepsilon$ falls below the bound we assume in the following theorem statement, forming $O(r)$ linear pieces will be nearly optimal.
\begin{theorem}
	\label{thm:tight}
	Let $\varepsilon > 0$ and $\vw(i) = \sqrt{i}$ be the square root $r$-SCB integer function. If $\varepsilon \geq r^{-\delta}$ for some constant $\delta \in (0,2)$, then any piecewise linear function providing a $(1+\varepsilon)$-approximation for $\vw$ contains $\Omega( \log_{\gamma(\varepsilon)} r)$ linear pieces, which behaves as $\Omega (\varepsilon^{-1/2} \log r)$ as $\varepsilon \rightarrow 0$.
\end{theorem}
\begin{proof}
	Let $L^*$ be the optimal set of linear pieces returned by running Algorithm~\ref{alg:findbest}. In order to show $|L^*| = \Omega( \log_{\gamma(\varepsilon)} r)$, we will construct a new set of linear pieces $L$ that has asymptotically the same number of linear pieces as $L^*$, but also provides a $(1+\varepsilon)$-approximation for all $x$ in an interval $[r^\beta, r]$ for some constant $\beta < 1$. Invoking Lemma~\ref{thm:magnanti} will then guarantee the final result.
	
	Recall that $L^*$ includes only two types of linear pieces: either linear pieces $g$ satisfying $g(j) = \sqrt{j}$ for exactly one integer $j$ (see~\eqref{onept}), or linear pieces formed by joining two points of $\vw$ (see~\eqref{gi}). For the square root splitting function, the latter type of linear piece is of the form
	\begin{equation}
	\label{eq:gt}
	g^{(t)}(i) = (\sqrt{t+1} - \sqrt{t}) (i - t) + \sqrt{t},
	\end{equation}
	for some positive integer $t$ less than $r$. This is the linear interpolation of the points $(t,\sqrt{t})$ and $(t+1, \sqrt{t+1})$. Both types of linear pieces bound $\phi(x) = \sqrt{x}$ above at \emph{integer} points, but they may cross below $\phi$ at non-integer values of $x$. To apply Lemma~\ref{thm:magnanti}, we would like to obtain a set of linear pieces that are all tangent lines to $\phi$. We accomplish this by replacing each linear piece in $L^*$ with two or three linear pieces that are tangent to $\phi$ at some point. For a positive integer $j$, let $g_j$ denote the line tangent to $\phi(x) = \sqrt{x}$ at $x = j$, which is given by
	\begin{equation}
	\label{eq:gj}
	g_j(x) = \frac{1}{2\sqrt{j}} (x - j) + \sqrt{j}.
	\end{equation}
	We form a new set of linear pieces $L$ made up of lines tangent to $\phi$ using the following replacements:
	\begin{itemize}
		\item If $L^*$ contains a linear piece $g$ that satisfies $g(j) = \sqrt{j}$ for exactly one integer $j$, add lines $g_{j-1}$, $g_j$, and $g_{j+1}$ to $L$. 
		\item If for an integer $t$, $L^*$ contains the line $g^{(t)}$ as given by Eq.~\eqref{eq:gt}, add lines $g_t$ and $g_{t+1}$ to $L$.
	\end{itemize}
	By Observation~\ref{obs:gh}, this replacement can only improve the approximation guarantee at integer points. Therefore,  $L$ provides a $(1+\varepsilon)$-approximation at integer values, is made up strictly of lines that are tangent to $\phi$, and contains at most three times the number of lines in $L^*$. 
	
	Due to the concavity of $\phi$, if a single line $g \in L$ provides a $(1+\varepsilon)$-approximation at consecutive integers $i$ and $i+1$, then $g$ provides the same approximation guarantee for all $x \in [i, i+1]$. 
	However, if two integers $i$ and $i+1$ are not \emph{both} covered by the \emph{same} line in $L$, then $L$ does not necessarily provide a $(1+\varepsilon)$-approximation for every $x \in [i,i+1]$. There can be at most $|L|$ intervals of this form, since each interval defines an ``intersection" at which one line $g \in L$ ceases to be a $(1+\varepsilon)$-approximation, and another line $g' \in L$ ``takes over'' as the line providing the approximation.
	
	By Lemma~\ref{thm:magnanti}, we can cover an entire interval $[i,i+1]$ for any integer $i$ using a set of $\lceil \log_{\gamma(\varepsilon)} \big(1 + \frac{1}{i} \big) \rceil $ linear pieces that are tangent to $\phi$ somewhere in $[i, i+1]$. Since $1+ \sqrt{\varepsilon} \leq \gamma(\varepsilon)$, it in fact takes only one linear piece to cover $[i, i+1]$ as long as $1+ 1/i \leq 1 + \sqrt{\varepsilon} \implies i \geq 1/\sqrt{\varepsilon}$. Since $\varepsilon \geq r^{-\delta}$, interval $[i, i+1]$ can be covered by a single linear piece if $i \geq r^{\delta/2}$. Therefore, for each interval $[i, i+1]$, with $i \geq r^{\delta/2}$, that is not already covered by a single linear piece in $L$, we add one more linear piece to $L$ to cover this interval. This at most doubles the size of $L$.
			
	The resulting set $L$ will have at most $6$ times as many linear pieces as $L^*$, and is guaranteed to provide a $(1+\varepsilon)$-approximation for all integers, as well as the entire continuous interval $[r^{\delta/2}, r]$. 
	Since $\delta$ is a fixed constant strictly less than 2, applying Lemma~\ref{thm:magnanti} shows that $L$ has at least
	\begin{equation*}
	\ceil*{\log_{\gamma(\varepsilon)} \frac{ r}{r^{\delta/2}} } = \Omega (\log_{\gamma(\varepsilon)} r^{1-\delta/2} ) = \Omega (\log_{\gamma(\varepsilon)} r)
	\end{equation*}
	 linear pieces. Therefore, $|L^*| = \Omega (\log_{\gamma(\varepsilon)} r)$ as well.
\end{proof}

\subsection{Improved Bound for the Clique Function}
When approximating the clique expansion splitting function, Algorithm~\ref{alg:findbest} will in fact find a piecewise linear curve with at most $O(\varepsilon^{-1/2}\log \log \frac{1}{\varepsilon})$ linear pieces. We prove this by highlighting a different approach for constructing a piecewise linear curve with this many linear pieces, which upper bounds the number of linear pieces in the \emph{optimal} curve found by Algorithm~\ref{alg:findbest}. 

Clique splitting penalties for a $k$-node hyperedge correspond to nonnegative integer values of the continuous function $\zeta(x) = x \cdot (k- x)$. 
As we did in Section~\ref{sec:findbest}, we want to build a set of linear pieces $L$ that provides and upper bounding $(1+\varepsilon)$-cover of $\zeta$ at integer values in $[0, r]$, where $r = \floor*{k/2}$. We start by adding the line $g^{(0)}(x) = (\vw(1) - \vw(0))x + \vw(0) = (k-1)\cdot x$ to $L$, which perfectly covers the first two splitting penalties $\vw(0) = 0$ and $\vw(1) = k-1$. In the remainder of our new procedure we will find a set of linear pieces to $(1+\varepsilon)$-cover $\zeta$ at \emph{every} value of $x \in [1,k/2]$, even non-integer $x$. 

We apply a greedy procedure similar to Algorithm~\ref{alg:findbest}. At each iteration we consider a leftmost endpoint $z_i$ which is the largest value in $[1,k/2]$ for which we already have a $(1+\varepsilon)$-approximation. In the first iteration, we have $z_1 = 1$. We then would like to find a new linear piece that provides a $(1+\varepsilon)$-approximation for all values from $z_i$ to some $z_{i+1}$, where the value of $z_{i+1}$ is maximized. We restrict to linear pieces that are tangent to $\zeta$. The line tangent to $\zeta$ at $t \in [1,k/2]$ is given by
\begin{equation}
\label{eq:linet}
g_t(x) = k x - 2tx + t^2 \,.
\end{equation}
We find $z_{i+1}$ in two steps:
\begin{enumerate}
	\item \textbf{Step 1:} Find the maximum value $t$ such that $g_t(z_i) = (1+\varepsilon) \zeta(z_i)$.
	\item \textbf{Step 2:} Given $t$, find the maximum $z_{i+1}$ such that $g_t(z_{i+1}) = (1+\varepsilon) \zeta(z_{i+1})$. 
\end{enumerate}
After completing these two steps, we add the linear piece $g_t$ to $L$, knowing that it covers all values in $[z_i, z_{i+1}]$ with a $(1+\varepsilon)$-approximation. At this point, we will have a cover for all values in $[0,z_{i+1}]$, and we begin a new iteration with $z_{i+1}$ being the largest value covered. We continue until we have covered all values up until $z_{i+1} \geq k/2$. If $t > k/2$ in Step 1 of the last iteration, we adjust the last linear piece to instead be the line tangent to $\zeta$ at $x = k/2$, so that we only include lines that have a nonnegative slope.
\begin{lemma}
	\label{lem:tzi}
	For any $z_i \in [1,k/2]$, the values of $t$ and $z_{i+1}$ given in steps 1 and 2 are given by
	\begin{align}
	\label{t}
	t &= z_i + \sqrt{ z_i (k - z_i) \varepsilon } \\
	\label{zi1}
	z_{i+1} &= \frac{t}{1+\varepsilon} + \frac{k \varepsilon}{2(1+\varepsilon)} + \frac{1}{2(1+\varepsilon)} \left(k^2 \varepsilon^2 + 4\varepsilon t (k - t)\right)^{1/2}
	\end{align}
\end{lemma}
\begin{proof}
	The proof simply requires solving two different quadratic equations. For Step 1:
	\begin{align*}
	g_t(z_i) = (1+\varepsilon) \zeta(z_i) &\iff k z_i - 2t z_i + t^2 = (1+\varepsilon)(z_i k - z_i^2) \\
	&\iff t^2 - 2 z_i t - \varepsilon z_i k + (1+\varepsilon)z_i^2 = 0
	\end{align*}
	Taking the larger solution to maximize $t$:
	\begin{equation*}
	t = \frac{1}{2} \left(2z_i + \sqrt{ 4z_i^2 - 4(1+\varepsilon) z_i^2 + 4 \varepsilon k z_i} \right) = z_i + \sqrt{z_i (k-z_i) \varepsilon }.
	\end{equation*}
	For Step 2:
	\begin{align*}
	g_t(z_{i+1}) = (1+\varepsilon) \zeta(z_{i+1}) & \iff k z_{i+1} - 2t z_{i+1} + t^2 = (1+\varepsilon)(z_{i+1}k - z_{i+1}^2) \\
	& \iff (1+\varepsilon) z_{i+1}^2  + z_{i+1} (-\varepsilon k - 2t ) + t^2 = 0.
	\end{align*}
	We again take the larger solution to this quadratic equation since we want to maximize $z_{i+1}$:
	\begin{align*}
	z_{i+1} &= \frac{1}{2(1+\varepsilon)} \left(\varepsilon k + 2t + \sqrt{\varepsilon^2 k^2 + 4t\varepsilon k + 4t^2 - 4(1+\varepsilon) t^2}  \right) \\
	&= \frac{1}{2(1+\varepsilon)}  \left(\varepsilon k + 2t + \sqrt{\varepsilon^2 k^2 + 4t\varepsilon(k-t)}\right).
	\end{align*}
\end{proof}
Algorithm~\ref{alg:clique} summarizes the new procedure for covering the clique splitting function. Since  $z_1 = 1$, if $\varepsilon \geq 1$, then
\[
z_2 \geq \frac{1}{2(1+\varepsilon)} (2 k\varepsilon) = \frac{k \varepsilon}{1+\varepsilon} \geq \frac{k}{2}, 
\]
so after one step we have covered the entire interval $[1,k/2]$. We can therefore focus on $\varepsilon < 1$.
\begin{algorithm}[t]
	\caption{Find a $(1+\varepsilon)$-cover $L$ for the clique splitting function.}
	\label{alg:clique}
	\begin{algorithmic}
		\State \textbf{Input:} Hyperedge size $k$, $\varepsilon \geq 0$
		\State \textbf{Output:} $(1+\varepsilon)$ cover for clique splitting function.
		\State $L = \{g^{(0)}\}$, where $g^{(0)}(x) = (k-1)x$
		\State $z = 1$
		\Do
		\State $t \leftarrow z + \sqrt{z(k-z)\varepsilon}$
		\State $z \leftarrow \frac{t}{1+\varepsilon} + \frac{k \varepsilon}{2(1+\varepsilon)} + \frac{1}{2(1+\varepsilon)} \left(k^2 \varepsilon^2 + 4\varepsilon t (k - t)\right)^{1/2}$
		\If{$t > k/2$} 
		\State $L \leftarrow L \cup \{g_{k/2}\}$, where $g_{k/2}(x) = k/2$
		\Else
		\State $L \leftarrow L \cup \{g_t\}$, where $g_t(x) = kx - 2tx + t^2$
		\EndIf
		\doWhile{$z_{i+1} < k/2$}
		\State Return $\vf$ defined by $\vf(x) = \min_{g \in L} g(x)$
	\end{algorithmic}
\end{algorithm}

\begin{theorem}
	\label{thm:clique}
	For $\varepsilon < 1$, if $L$ is the output from Algorithm~\ref{alg:clique}, then $|L| = O(\varepsilon^{-1/2} \log \log \frac{1}{\varepsilon})$.
\end{theorem}
\begin{proof}
	We get a loose bound for the value of $t$ in Lemma~\ref{lem:tzi} by noting that $(k - z_i) \geq k/2 \geq z_i$:
	\begin{equation}
	\label{tt1}
	t = z_i + \sqrt{z_i \varepsilon (k - z_i)} \geq z_i + \sqrt{z_i^2 \varepsilon} = z_i (1+ \sqrt{\varepsilon}).
	\end{equation}
	Since we assumed $\varepsilon < 1$, we know that 
	\begin{equation}
	\frac{t}{1+\varepsilon} \geq \frac{ z_i (1 + \sqrt{\varepsilon})}{1+\varepsilon} > z_i.
	\end{equation}
	Therefore, from~\eqref{zi1} we see that
	\begin{align}
	\label{boundzi}
	z_{i + 1} &> z_i + \frac{k \varepsilon}{2(1+\varepsilon)} + \frac{1}{2(1+\varepsilon)} \left(k^2 \varepsilon^2 + 4\varepsilon t (k - t)\right)^{1/2} \\
	&> z_i + \frac{k \varepsilon}{2(1+\varepsilon)}  + \frac{1}{2(1+\varepsilon)} \left(k^2 \varepsilon^2\right)^{1/2} = z_i + \frac{k\varepsilon}{1+\varepsilon}.
	\end{align}
	From this we see that at each iteration, we cover an additional interval of length $z_{i+1} - z_i > \frac{k\varepsilon}{1+\varepsilon}$, and therefore we know it will take at most $O(1/\varepsilon)$ iterations to cover all of $[1,k/2]$. This upper bound is loose, however. The value of $z_{i+1} - z_i$ in fact increases significantly with each iteration, allowing the algorithm to cover larger and larger intervals as it progresses. 
	
	Since $z_1 = 1$ and $z_{i+1} - z_i \geq \frac{k\varepsilon}{1+\varepsilon}$, we see that $z_j \geq k\varepsilon$ for all $j \geq 3$. For the remainder of the proof, we focus on bounding the number of iterations it takes to cover the interval $[k\varepsilon, k/2]$. 
	We separate the progress made by Algorithm~\ref{alg:clique} into different rounds. Round $j$ refers to the set of iterations that the algorithm spends to cover the interval 
	\begin{equation}
	\label{roundj}
	R_j = \left[k \varepsilon^{\left(\frac{1}{2}\right)^{j-1}}, k \varepsilon^{\left(\frac{1}{2}\right)^{j}}\right] ,
	\end{equation}
	For example, Round 1 starts with the iteration $i$ such that $z_i \geq k\varepsilon$, and terminates when the algorithm reaches an iteration $i'$ where $z_{i'} \geq k\varepsilon^{1/2}$. A key observation is that it takes less than $4/\sqrt{\varepsilon}$ iterations for the algorithm to finish Round $j$ for any value of $j$. To see why, observe that from the bound in~\eqref{boundzi} we have
		\begin{align*}
	z_{i + 1} - z_i &>  \frac{k \varepsilon}{2(1+\varepsilon)} + \frac{1}{2(1+\varepsilon)} \left(k^2 \varepsilon^2 + 4\varepsilon t (k - t)\right)^{1/2} \\
	&>\frac{1}{2(1+\varepsilon)} \left(4\varepsilon t (k - t)\right)^{1/2} \\
	&\geq   \frac{1}{2(1+\varepsilon)} \left(4\varepsilon z_i \frac{k}{2}\right)^{1/2} \\
	&> \frac{\sqrt{2}}{2} \frac{\sqrt{k\varepsilon}}{(1+\varepsilon)} \sqrt{z_i}.
	\end{align*}
	For each iteration $i$ in Round $j$, we know that $z_i \geq k \varepsilon^{\left(\frac{1}{2}\right)^{j-1}}$, so that
	\begin{equation}
	z_{i+1} - z_i > \frac{\sqrt{2}}{2} \frac{\sqrt{k\varepsilon}}{(1+\varepsilon)}  \sqrt{k \varepsilon^{\left(\frac{1}{2}\right)^{j-1}}} \geq \frac{\sqrt{2}}{2} \frac{k \varepsilon^{\frac{1}{2} + \left(\frac{1}{2}\right)^{j} }}{1+\varepsilon} 
	= C \cdot k \cdot \varepsilon^{\frac{1}{2} + \left(\frac{1}{2}\right)^{j} },
	\end{equation}
	where $C = \sqrt{2}/(2(1+\varepsilon))$ is a constant larger than $1/4$.
	Since each iteration of Round $j$ covers an interval of length at least $C\cdot k \cdot \varepsilon^{\frac{1}{2} + \left(\frac{1}{2}\right)^{j} }$, and the right endpoint for Round $j$ is $k \varepsilon^{\left(\frac{1}{2}\right)^{j}}$, the maximum number of iterations needed to complete Round $j$ is
	\begin{equation}
	\frac{ k \varepsilon^{\left(\frac{1}{2}\right)^{j}}} {C \cdot k \cdot \varepsilon^{\frac{1}{2} + \left(\frac{1}{2}\right)^{j} }} = \frac{1}{C \sqrt{\varepsilon}}.
	\end{equation}
	Therefore, after $p$ rounds, the algorithm will have performed $O(p \cdot \varepsilon^{-1/2})$ iterations, to cover the interval $[1,k \varepsilon^{\left(\frac{1}{2}\right)^{p}}]$. Since we set out to cover the interval $[1, k/2]$, this will be accomplished as soon as $p$ satisfies $\varepsilon^{\left(\frac{1}{2}\right)^{p}} \geq 1/2$, which holds as long as $p \geq \log_2 \log_2 \frac{1}{\varepsilon}$:
	\begin{align*}
	\varepsilon^{\left(\frac{1}{2}\right)^{p}} \geq 1/2 
	& \iff \left(\frac{1}{2}\right)^{p} \log_2 \varepsilon \geq -1 \\
	& \iff \log_2 \varepsilon \geq -2^p \\
	& \iff \log_2 \frac{1}{\varepsilon} \leq 2^p \\
	& \iff \log_2 \log_2\frac{1}{\varepsilon} \leq p .
	\end{align*}
	This means that the number of iteration of Algorithm~\ref{alg:clique}, and therefore the number of linear pieces in $L$, is bounded above by $O(\varepsilon^{-1/2} \log \log \frac{1}{\varepsilon})$.
\end{proof}
We obtain a proof of Theorem~\ref{thm:complete} on sparsifying the complete graph as a corollary.

\paragraph{Proof of Theorem~\ref{thm:complete}.}
A complete graph on $n$ nodes can be viewed as a hypergraph with a single $n$-node hyperedge with a clique expansion splitting function. Theorem~\ref{thm:complete} says that the clique expansion integer function $\vw(i) = i \cdot (n-i)$ can be covered with $O(\varepsilon^{-1/2} \log \log \varepsilon^{-1})$ linear pieces, which is equivalent to saying the clique expansion splitting function can be modeled using this many CB-gadgets. Each CB-gadget has two auxiliary nodes and $(2n+1)$ directed edges. This results in an augmented sparsifier for the complete graph with $O(n \varepsilon^{-1/2} \log \log \varepsilon^{-1})$ edges. This is only meaningful if $\varepsilon$ is small enough so that $O(\varepsilon^{-1/2} \log \log \varepsilon^{-1})$ is asymptotically less than $n$, so our sparsifier has $O(n + \varepsilon^{-1/2} \log \log \varepsilon^{-1}) = O(n)$ nodes.

\hfill $\square$

\section{Sparsifying Co-occurrence Graphs}
\label{sec:ccs}
Recall from the introduction that a co-occurrence graph is formally defined by a
set of nodes $V$ and a set of subsets $\mathcal{C} \subseteq 2^V$. In practice,
each $c \in \mathcal{C}$ could represent some type of group interaction
involving nodes in $c$ or a set of nodes sharing the same attribute.
We define
the co-occurrence graph $G = (V,E)$ on $\mathcal{C}$ to be the graph where nodes
$i$ and $j$ share an edge with weight $w_{ij} = \sum_{c \in \mathcal{C}} w_c$,
where $w_c \geq 0$ is a weight associated with co-occurrence set $c \in \mathcal{C}$.
The case when $w_c = 1$ is standard and is an example of a common practice
of ``one-mode projections''  of bipartite graphs or affiliation
networks~\cite{lattanzi2009affiliation,li2007evolving,neal2014backbone,newman2001random,ramasco2006social,zhou2007bipartite,benson2020simple} --- a graph is formed on the nodes from one side of a bipartite graph by
connecting two nodes whenever they share a common neighbor on the other side,
where edges are weighted based on the number of shared neighbors.

%
A co-occurrence graph $G$ has the following \emph{co-occurrence} cut function:
\begin{equation}
\label{eq:cocut2}
\cut_G(S) = \sum_{c \in \mathcal{C}} w_c \cdot |S \cap c| \cdot | \bar{S} \cap c|.
\end{equation}
In this sense, the co-occurrence graph is naturally interpreted as a weighted clique expansion of a hypergraph $\mathcal{H} = (V, \mathcal{C})$, which itself is a special case of reducing a submodular, cardinality-based hypergraph to a graph. However, this type of graph construction is by no means restricted to literature on hypergraph clustering. In many applications, the first step in a larger experimental pipeline is to construct a graph of this type from a large dataset. The resulting graph is often quite dense,
as numerous domains involve large hyperedges~\cite{veldt2020hyperlocal,purkait2016clustering}. This makes it expensive to form, store, and compute over co-occurrence graphs in practice.

Solving cut problems on these dense co-occurrence graphs arises naturally in many settings.
For example, any hypergraph clustering application that relies on a clique expansion involves a graph with a co-occurrence cut function~\cite{Agarwal2006holearning,hadley1995,Hein2013,Huang2015,Zhou2006learning,panli2017inhomogeneous,vannelli1990Gomoryhu,zien1999,rodriguez2009laplacian,veldt2020paramcc,veldt2020hyperlocal}.
Clustering social networks is another use case, as
online platforms have many ways to create groups of users (e.g., events, special interest groups, businesses, organizations, etc.), that can be large in practice.
Furthermore, cuts in co-occurrence graphs of students on a university campus (based on, e.g., common classes, living arrangements, or physical proximity)
are relevant to preventing the spread of infectious diseases such as COVID-19.


In these cases, it would be  more efficient to sparsify the graph \emph{without ever forming it explicitly}, by sparsifying large cliques induced by co-occurrence relationships. 
Although this strategy seems intuitive, it is often ignored in practice. 
We therefore present several theoretical results that highlight the benefits of this implicit approach to sparsification.
Our focus is on results that can be achieved using augmented sparsifiers for cliques, though many of the same benefits could also be achieved with standard sparsification techniques.

\subsection{Proof of Theorem~\ref{thm:cooc}}
Let $\mathcal{C}$ be a set of nonempty co-occurrence groups on a set of $n$ nodes, $V$, and let $G = (V,E)$ be the corresponding co-occurrence graph on $\mathcal{C}$. For $c \in \mathcal{C}$, let $k_c = |c|$ be the number of nodes in $c$. For $v \in V$, let $d_v$ be the co-occurrence degree of $v$: the number of sets $c$ containing $v$. Let $d_\mathit{avg} = \frac{1}{n} \sum_{v \in V} d_v$ be the average co-occurrence degree. We re-state and prove Theorem~\ref{thm:cooc}, first presented in the introduction. The proof holds independent of the weight $w_c$ we associate with each $c \in \mathcal{C}$, since we can always scale our graph reduction techniques by an arbitrary positive weight.
\begin{theoremintro}
 Let $\varepsilon > 0$ and $f(\varepsilon) = \varepsilon^{-1/2}\log \log \frac{1}{\varepsilon}$. There exists an augmented sparsifier for $G$\ with $O(n + |\mathcal{C}| \cdot f(\varepsilon))$ nodes and $O(n\cdot d_\mathit{avg} \cdot f(\varepsilon) )$ edges. In particular, if $d_{avg}$ is constant and for some $\delta > 0$ we have $\sum_{c \in \mathcal{C}} |c|^2 = \Omega(n^{1+\delta})$, then forming $G$ explicitly takes $\Omega(n^{1+\delta})$ time, but an augmented sparsifier for $G$ with $O(nf(\varepsilon))$ nodes and $O(n f(\varepsilon))$ edges can be constructed in $O( n f(\varepsilon))$ time.
\end{theoremintro}
\begin{proof}
	The set $c$ induces a clique in the co-occurrence graph with $O(k_c^2)$ edges. Therefore, the runtime for explicitly forming $G = (V,E)$ by expanding cliques and placing all edges equals $O(\sum_{c \in \mathcal{C}} k_c^2) = \Omega(n^{1+\delta})$. By Theorem~\ref{thm:complete}, for each $c \in \mathcal{C}$ we can produce an augmented sparsifier with $O(k_c f(\varepsilon))$ directed edges and $O( f(\varepsilon))$ new auxiliary nodes. Sparsifying each clique in this way will produce an augmented sparsifier $\hat{G} = (\hat{V}, \hat{E})$ where
	\begin{align}
	\label{ehat}
	|\hat{E}|& = \sum_{c \in \mathcal{C}} O(k_c f(\varepsilon)) = O(f(\varepsilon) \cdot n \cdot d_\mathit{avg} )\\
	\label{vhat}
	|\hat{V}| &= n + \sum_{c \in \mathcal{C}} O(f(\varepsilon)) = O(n + |\mathcal{C}| f(\varepsilon)).
	\end{align}
	Observe that $n \cdot d_\mathit{avg} = \sum_{v \in V}  d_v = \sum_{c \in \mathcal{C}} k_c$.  If $d_{avg}$ is a constant, this implies that $\sum_{c \in \mathcal{C}} k_c = O(n)$, and furthermore that $|\mathcal{C}| = O(n)$, since each $k_c \geq 1$. Therefore $|\hat{E}|$ and $|\hat{V}|$ are both $O(n f(\varepsilon))$. Only $O(f(\varepsilon))$ edge weights need to be computed for the clique, so the overall runtime is just the time it takes to explicitly place the $O(n f(\varepsilon))$ edges.
\end{proof}
The above theorem and its proof includes the case where $|\mathcal{C}|= o(n)$, meaning that $\mathcal{C}$ is made up of a sublinear number of large co-occurrence interactions. In this case, our augmented sparsifier will have fewer than $O(n f(\varepsilon))$ nodes. When $|\mathcal{C}| = \omega(n)$, the average degree will no longer be a constant and therefore it becomes theoretically beneficial to sparsify each clique in $\mathcal{C}$ using standard undirected sparsifiers. For each $c \in \mathcal{C}$, standard cut sparsification techniques will produce an $\varepsilon$-cut sparsifier of $c$ with $O(k_c \varepsilon^{-2})$ undirected edges and exactly $k_c$ nodes. If two nodes appear in multiple co-occurrence relationships, the resulting edges can be collapsed into a weighted edge between the nodes, meaning that the number of edges in the resulting sparsifier does not depend on $d_\mathit{avg}$. We discuss tradeoffs between different sparsification techniques in depth in a later subsection. Regardless of the sparsification technique we apply in practice, implicitly sparsifying a co-occurrence graph will often lead to a significant decrease in runtime compared to forming the entire graph prior to sparsifying it. 

\subsection{A Simple Co-occurrence Model}
We now consider a simple model for co-occurrence graphs with a power-law group
size distribution, that produces graphs satisfying the conditions of
Theorem~\ref{thm:cooc} in a range of different parameter settings.  Such
distributions have been observed for many types of co-occurrence graphs
constructed from real-world data~\cite{clauset2009powerlaw,benson2020simple}.
More formally, let $V$ be a set of $n$ nodes, and assume a co-occurrence set $c$
is randomly generated by sampling a set of size $K$ from a discrete power-law
distribution where for $k \in [1,n]$:
\begin{equation*}
\mathbb{P}[K = k] = Ck^{-\gamma} .
\end{equation*}
Here, $C$ is a normalizing constant for the distribution, and $\gamma$ and is a parameter of the model. Once $K$ is drawn from this model, a co-occurrence set $c$ is generated with a set of $K$ nodes from $V$ chosen uniformly at random.
This procedure can be repeated an arbitrary number of times (drawing all sizes $K$ independently) to produce a set of co-occurrence sets $\mathcal{C}$. This $\mathcal{C}$ can then be used to generate a co-occurrence graph $G = (V,E)$.
(The end result of this procedure is a type of \emph{random intersection graph}~\cite{bloznelis2013degree}.)
We first consider a parameter regime where set sizes are constant on average but large enough to produce a dense co-occurrence graph that is inefficient to explicitly form in practice.
The regime has an exponent $\gamma \in (2, 3)$, which is common in real-world data~\cite{clauset2009powerlaw}.
\begin{theorem}
	Let $\mathcal{C}$ be a set of $O(n)$ co-occurrence sets obtained from the power-law model with $\gamma \in (2,3)$. The expected degree of each node will be constant and
$\mathbb{E}\left [ \sum_{c \in \mathcal{C}} |c|^2\right] = O(n^{4-\gamma})$.
\end{theorem}
\begin{proof}
	Let $K$ be the size of a randomly generated co-occurrence set. We compute:
		\begin{equation*}
	\mathbb{E}[K^2] = \sum_{k = 1}^n k^2 \cdot \mathbb{P}[K = k]  =  C \cdot \sum_{k = 1}^n k^{2-\gamma} \leq C \cdot \left[1 + \int_{1}^n x^{2-\gamma} dx \right]= C+ \frac{C n^{3-\gamma}}{3-\gamma} - \frac{C}{3-\gamma}  = O(n^{3- \gamma}).
	\end{equation*}
	Therefore,
	\begin{align*}
	\mathbb{E}\left [ \sum_{c \in \mathcal{C}} |c|^2\right] &=\sum_{c \in \mathcal{C}} \mathbb{E}[K^2] = O(n^{4 - \gamma}).
	\end{align*}
	For a node $v \in V$ and a randomly generated set $c$, the probability that $v$ will be selected to be in $c$ is
	\begin{equation*}
	\mathbb{P}[v \in c] = \sum_{k = 1}^n \mathbb{P}[|c| = k] \cdot \frac{{n-1 \choose k-1}}{{n \choose k} } = C \cdot \sum_{k = 1}^n k^{-\gamma} \cdot \frac{k}{n} =\frac{C}{n} \cdot \left[ 1 + \int_{1}^n x^{1-\gamma} dx \right] = O(n^{-1}).
	\end{equation*}
	Since there are $O(n)$ co-occurrence sets in $\mathcal{C}$ and they each are generated independently, in expectation, $v$ will have a constant degree.
\end{proof}

We similarly consider another regime of co-occurrence graphs where the number of co-occurrence sets is asymptotically smaller than $n$, but the co-occurrence sets are larger on average.
\begin{theorem}
	\label{thm:largeclique}
	Let $\mathcal{C}$ be a set of $O(n^\beta)$ co-occurrence sets, where $\beta \in (0,1)$, obtained from the power-law co-occurrence model with $\gamma = 1 + \beta$. Then the expected degree of each node will be a constant and
	$\mathbb{E}\left [ \sum_{c \in \mathcal{C}} |c|^2\right] = O(n^{3})$.
\end{theorem}
\begin{proof}
	Again let $K$ be a random variable representing the co-occurrence set size. We have
	\begin{align*}
	\mathbb{E}[K^2] = C \cdot \sum_{k = 1}^n k^{2-\gamma}  = O(n^{3- \gamma}) \implies \mathbb{E}\left [ \sum_{c \in \mathcal{C}} |c|^2\right] = O(n^{\beta + 4 - \gamma}) = O(n^3).
	\end{align*}
	For a node $v \in V$ and a randomly generated set $c$, the probability that $v$ will be in $c$ is
	\begin{equation*}
	\mathbb{P}[v \in c] = \sum_{k = 1}^n \mathbb{P}[|c| = k] \cdot \frac{{n-1 \choose k-1}}{{n \choose k} } = \frac{C}{n} \sum_{k = 1}^n k^{1-\gamma} = O(n^{2-\gamma-1}) = O(n^{-\beta})
	\end{equation*}
	Since there are $O(n^\beta)$ co-occurrence sets in $\mathcal{C}$, the expected degree of $v$ is a constant. 
\end{proof}
In Theorem~\ref{thm:largeclique}, the exponent of the power-law distribution is assumed to be directly related to the number of co-occurrence sets in $\mathcal{C}$. This assumption is included simply to ensure that we are in fact considering co-occurrence graphs with $O(n)$ nodes. We could alternatively consider a power-law distribution with exponent $\gamma \in (1,2)$ and generate $O(n^\beta)$ co-occurrence sets for any $\beta < 1-\gamma$. We simply note that in this regime, the expected average degree will be $o(1)$. Assuming we exclude isolated nodes, this will produce a co-occurrence graph with $o(n)$ nodes in expectation. Our techniques still apply in this setting, and we can produce augmented sparsifiers with $O(|\mathcal{C}| \cdot f(\varepsilon))$ nodes and $O(n \cdot d_\mathit{avg} \cdot f(\varepsilon)) = o(n \cdot f(\varepsilon))$ edges. When $|\mathcal{C}| = \Omega(n)$, then $d_\mathit{avg} = \Omega(1)$ and the number of edges in our augmented sparsifiers will have worse than linear dependence on $n$. However, in this regime we can still quickly obtain sparsifiers with $O(n\varepsilon^{-2})$ edges via implicit sparsification by using standard undirected sparsifiers.

More sophisticated models for generating co-occurrence graphs can also be derived from existing models for projections of bipartite graphs~\cite{benson2020simple,bloznelis2013degree,bloznelis2017correlation}. These make it possible to set different distributions for node degrees in $V$  and highlight other classes of co-occurrence graphs satisfying the assumptions of Theorem~\ref{thm:cooc}. Here we have chosen to focus on the simplest model for illustrating classes of power-law co-occurrence graphs that satisfy the assumptions of the theorem.

\subsection{Tradeoffs in sparsification techniques}
There are several tradeoffs to consider when using different black-box sparsifiers for implicit co-occurrence sparsification. Standard sparsification techniques involve no auxiliary nodes, and have undirected edges, which is beneficial in numerous applications. Also, the number of edges they require is independent of $d_\mathit{avg}$. Therefore, in cases where the average co-occurrence degree is larger than a constant, we obtain better theoretical improvements using standard sparsifiers.

On the other hand, in many settings, it is natural to assume the number of co-occurrences each node belongs to is a constant, even if some co-occurrences are very large. In these regimes, our augmented sparsifiers will have fewer edges that traditional sparsifiers due to a better dependence on $\varepsilon$. Our techniques are also deterministic and our sparsifiers are very easy to construct in practice. Edge weights for our sparsifiers are easy to determine in $O(f(\varepsilon))$ time for each co-occurrence group using Algorithm~\ref{alg:findbest} (or Algorithm~\ref{alg:clique}) coupled with Lemma~\ref{lem:pl2ccb}. The bottleneck in our construction is simply visiting each node in a set $c$ to place edges between it and the auxiliary nodes. Even in cases where there are no asymptotic reductions in theoretical runtime, our techniques provide a simple and highly practical tool for solving cut problems on co-occurrence data.


\section{Approximate Cardinality-based DSFM}
\label{sec:runtime}
Appendix~\ref{app:general} shows how our sparse reduction techniques can be adjusted to apply even when splitting functions are asymmetric and are not required to satisfy $\vw_e(\emptyset) = \vw_e(e) = 0$ (the non-cut ignoring property). Section~\ref{sec:reductions} addresses the special case of symmetric and non-cut ignoring functions, as these assumptions are more natural for hypergraph cut problems~\cite{panli2017inhomogeneous,panli_submodular,veldt2020hypercuts}, and provide the clearest exposition of our main techniques and results. 
Furthermore, applying the generalized asymmetric reduction strategy in Appendix~\ref{app:general} to a symmetric splitting function would introduce twice as many edges as applying the reduction from Section~\ref{sec:reductions} designed explicitly for the symmetric case. Nevertheless, the same asymptotic upper bound of $O(\frac{1}{\varepsilon} \log k)$ edges holds for approximately modeling the more general splitting function on a $k$-node hyperedge. By dropping the symmetry and non-cut ignoring assumptions, our techniques lead to the first approximation algorithms for the more general problem of minimizing cardinality-based decomposable submodular functions.

\subsection{Decomposable Submodular Function Minimization}
Any submodular function can be minimized in polynomial time~\cite{Orlin2009,Iwata:2001:CSP:502090.502096,Iwata:2009:SCA:1496770.1496903}, but the runtimes for general submodular functions are impractical in most cases. A number of recent papers have developed faster algorithms for minimizing submodular functions that are sums of simpler submodular functions~\cite{ene2017decomposable,kolmogorov2012minimizing,stobbe2010efficient,panli2018revisiting,nishihara2014convergence,ene2015random,jegelka2011fast,jegelka2013reflection}. This is also know as decomposable submodular function minimization (DSFM). Many energy minimization problems from computer vision correspond to DSFM problems~\cite{kohli2009robust,kolmogorov2004what,freedman2005energy}.

Let $f\colon 2^V \rightarrow \mathbb{R}^+$ be a submodular function, such that for $S \subseteq V$,
\begin{equation}
f(S) = \sum_{e \in \mathcal{E}} \vf_e(S \cap e),
\end{equation}
where for each $e \in \mathcal{E}$, $\vf_e$ is a simpler submodular function with support only on a subset $e \subseteq V$.  We can assume without loss of generality that every $\vf_e$ is a non-negative function. The goal of DFSM is to find $\arg\min_S f(S)$. The terminology used for problems of this form differs depending on the context. We will continue to refer to $\mathcal{E}$ as a hyperedge set, $V$ as a node set, $\vf_e$ as generalized splitting functions, and $f$ as some type of generalized hypergraph cut function. 

Cardinality-based decomposable submodular function minimization (Card-DSFM) has received significant special attenation in previous work~\cite{kohli2009robust,kolmogorov2012minimizing,jegelka2011fast,jegelka2013reflection,stobbe2010efficient}. This variant of the problem assumes that each function $\vf_e$ satisfies $\vf_e(S) = g_e(|S|)$ for some concave function $g_e$  Unlike most existing work on generalized hypergraph cut functions~\cite{panli2017inhomogeneous,panli_submodular,veldt2020hypercuts}, research on DFSM does not typically assume that the functions $\vf_e$ are symmetric, and also do not assume that $\vf_e(\emptyset) = \vf_e(e) = 0$. In the appendix we demonstrate that our sparse reduction techniques can also be adapted to apply to this more general setting as well. This leads to a proof of Theorem~\ref{thm:card}.

\paragraph{Proof of Theorem~\ref{thm:card}}
Lemma~\ref{lem:upperlog} shows that every cardinality-based submodular function $\vf_e$ (which does not need to be symmetric or satisfy $\vf_e(e) = \vf_e(\emptyset) = 0$ ) can be approximately modeled by a combination of $O(\frac{1}{\varepsilon}\log k)$ asymmetric CB-gadgets. This implies that any instance of Card-DSFM can be reduced to an $s$-$t$ cut problem on a graph with $N = O(n + \frac{1}{\varepsilon}\sum_{e \in \mathcal{E}} \log |e|)$ nodes and $M = O(n+\frac{1}{\varepsilon}\sum_{e \in \mathcal{E}} \log |e|)$ edges. Applying a recent $\tilde{O}(M+N^{1.5})$-time algorithm for the maximum $s$-$t$ flow problem of van den Brand et al.~\cite{brand2021minimum} yields the overall runtime guarantee of $\tilde{O}\left( \frac{1}{\varepsilon} \sum_{e \in \mathcal{E} }|e| + (n + \frac{R}{\varepsilon})^{3/2} \right)$. 
If we wish to find the optimal solution to an instance of Card-DSFM, we can also set $\varepsilon = 0$ and obtain a reduced graph with $O(\sum_{e \in \mathcal{E} } |e|)$ nodes and $O(\sum_{e \in \mathcal{E} } |e|^2)$ edges. The resulting runtime will then be $\tilde{O}\left(\sum_{e \in \mathcal{E} } |e|^2 + \left(\sum_{e \in \mathcal{E} } |e|\right)^{3/2} \right)$. 

\subsection{Runtime Comparison for Cardinality-Based DFSM}
We refer to our approximation algorithm for Card-DSFM as \textsc{SparseCard}, since it depends on sparse reduction techniques. In the remainder of the section, we provide a careful runtime comparison between \textsc{SparseCard} and competing runtimes for Card-DSFM. We focus on each runtime's dependence on $n = |V|$, $ R = |E|$, and support sizes $|e|$, and use $\tilde{O}$ notation to hide logarithmic factors of $n$, $R$, and $1/{\varepsilon}$. To easily compare weakly polynomial runtimes, we assume that each $\vf_e$ has integer outputs, and assume that $\log (\max_S \, f(S))$ is small enough that it can also be absorbed by $\tilde{O}$ notation. Our primary goal is to highlight the runtime improvements that are possible when an approximate solution suffices. Among algorithms for DSFM, \textsc{SparseCard} is unique in its ability to quickly find solutions with a priori multiplicative approximation guarantees. Previous approaches for DSFM focus on either obtaining exact solutions, or finding a solution to within an \emph{additive} approximation error $\epsilon > 0$~\cite{axiotis2021decomposable,ene2017decomposable,panli2018revisiting,jegelka2013reflection}. In the latter case, setting $\epsilon$ small enough will guarantee an optimal solution in the case of integer output functions. However, these results provide no a priori multiplicative approximation guarantee, which is the traditional focus of approximation algorithms. Furthermore, using a larger value of $\epsilon$ for these additive approximations only improves runtimes in logarithmic terms. 
In constrast, setting $\varepsilon > 0$ will often lead to substantial runtime decreases for \textsc{SparseCard}.

\paragraph{Competing runtime guarantees.}
Table~\ref{tab:dfsm} lists runtimes for existing methods for DSFM. We have listed the asymptotic runtime for \textsc{SparseCard} when applying the recent maximum flow algorithm of van den Brand et al.~\cite{brand2021minimum}. While this leads to the best theoretical guarantees for our method, asymptotic runtime improvements over competing methods can also be shown using alternative fast algorithms for maximum flow~\cite{gao2021fully,lee2014path,goldberg1998beyond}.
%
%
For the submodular flow algorithm of Kolmogorov~\cite{kolmogorov2012minimizing}, we have reported the runtime guarantee provided specifically for Card-DSFM. While other approaches have frequently been applied to Card-DSFM~\cite{ene2017decomposable,stobbe2010efficient,panli2018revisiting,jegelka2013reflection}, runtimes guarantees for this case have not been presented explicitly and are more challenging to exactly pinpoint. Runtimes for most algorithms depend on certain oracles for solving smaller minimization problems at functions $\vf_e$ in an inner loop.
For $e \in E$, let $\mathcal{O}_e$ be a \emph{quadratic minimization oracle}, which for an arbitrary vector $w$ solves $\min_{y \in B(\vf_e)} \|y + w \|$ where $B(\vf_e)$ is the base polytope of the submodular function $\vf_e$ (see~\cite{ene2017decomposable,Bach:2013:LSF:2602000,jegelka2013reflection} for details). Let $\theta_e$ be the time it take to evaluate the oracle at $e \in E$, and define $\theta_\mathit{max} = \max_{e \in E} \theta_e$ and $\theta_\mathit{avg} = \frac{1}{R} \sum_{e \in E} \theta_e$. Although these oracles admit faster implementations in the case of concave cardinality functions, it is not immediately clear from previous work what is the best possible runtime. When $w =0$, solving $\min_{y \in B(\vf_e)} \|y + w \|$ takes $O(|e| \log |e|)$ time~\cite{jegelka2013reflection}, so this serves as a best case runtime we can expect for the more general oracle $\mathcal{O}_e$ based on previous results. We note also that in the case of the function function $\vf_e(A) = |A| |e\backslash A|$, Ene et al.~\cite{ene2017decomposable} highlight that an $O(|e| \log |e| + |e| \tau_e)$ algorithm can be used, where $\tau_e$ denotes the time it take to evaluate $\vf_e(S \cap e)$ for any $S \subseteq e$. In our runtime comparisons will use the bound $\theta_e = \Omega(|e|)$, as it is reasonable to expect that any meaningful submodular function we consider should take a least linear time to minimize.

\paragraph{Fast approximate solutions ($\varepsilon > 0$).} Barring the regime where support sizes $|e|$ are all very small, the accelerated coordinate descent method (ACDM) of Ene et al.~\cite{ene2015random} provides the fastest previous runtime guarantee. For a simple parameterized runtime analysis, consider a DSFM problem where the average support size is $({1}/{R}) \sum_e |e| = \Theta(n^{\alpha})$ for $\alpha \in [0,1]$, and $R = \Theta(n^\beta)$, where $\beta \geq 1-\alpha$ must hold if we assume each $v \in V$ is in the support for at least one function $\vf_e$. 
An exact runtime comparison between \textsc{SparseCard} and ACDM depends on the best runtime for the oracle $\mathcal{O}_e$ for concave cardinality functions. If an $O(|e| \log |e|)$ oracle is possible, the overall runtime guarantee for ACDM would be $\Omega(n^{1+\alpha+\beta})$. Meanwhile, for a small constant $\varepsilon > 0$, \textsc{SparseCard} provides a $(1+\varepsilon)$-approximate solution in time $\tilde{O}(n^{\alpha+\beta} + \max\{ n^{3/2}, n^{3\beta/2}\})$, which will faster by at least a factor $\tilde{O}(\sqrt{n})$ whenever $\beta \leq 1$. When $\beta > 1$, finding an approximation with \textsc{SparseCard} is guaranteed to be faster whenever $R = o(n^{2+2\alpha})$. 
If the best case oracle  $\mathcal{O}_e$ for concave cardinality functions is $\omega(|e| \log |e|)$, the runtime improvement of our method is even more significant. 

\paragraph{Guarantees for exact solutions ($\varepsilon = 0$).} As an added bonus, running \textsc{SparseCard} with $\varepsilon = 0$ leads to the fastest runtime for finding \emph{exact} solutions in many regimes. In this case, we can guarantee \textsc{SparseCard} will be faster than ACDM when the average support size is $\Theta(n^{\alpha})$ and $R = o(n^{2-\alpha})$. \textsc{SparseCard} can also find exact solutions faster than other discrete optimization methods~\cite{kolmogorov2012minimizing,axiotis2021decomposable,fix2013structured} in wide parameter regimes. 
Our method has a faster runtime guarantee than the submodular flow algorithm of Kolmogorov~\cite{kolmogorov2012minimizing} in all parameter regimes, and has a better runtime guarantee than the strongly-polynomial IBFS algorithm~\cite{fix2013structured}  when $\sum_{e} |e| = o(n^4)$, which includes all cases where $R = o(n^3)$. Both variants of IBFS as well as the recent method of Axiotis et al.~\cite{axiotis2021decomposable} become impractical if even a \emph{single} function $\vf_e$ has support size $|e| = O(n)$ (even when the average support size is much smaller). In this case, runtimes for these methods are $\Omega(n^5)$. Even using the very loose bounds $\sum_e |e|^2 \leq Rn^2$ and $\sum_e |e| \leq nR$, our method is guaranteed to find exact solutions faster as long as $R = o(n^{7/3})$, with significant additional improvements when approximate solutions suffice. In the extreme case where $\max_e |e| = O(1)$, the algorithm of Axiotis et al.~\cite{axiotis2021decomposable} obtains the best theoretical guarantees, as it has the same asymptotic runtime as a single maximum flow computation on an $n$-node, $R$-edge graph. It is worth noting that in this regime, running \text{SparseCard} with $\varepsilon  = 0$ can exactly solve Card-DSFM with the same asymptotic runtime guarantee as long as $R = \Theta(n)$, and has the added practical advantage that it requires only one call to a maximum flow oracle.

The runtime guarantee for \textsc{SparseCard} when $\varepsilon = 0$ can be matched asymptotically by combining existing exact reduction techniques~\cite{kohli2009robust,stobbe2010efficient,veldt2020hypercuts} with fast maximum flow algorithms. However, our method has the practical advantage of finding the \emph{sparsest} exact reduction in terms of CB-gadgets. For example, this results in a reduced graph with roughly half the number of edges required if we were to apply our previous exact reduction techniques~\cite{veldt2020hypercuts}. {Analogously, while Stobbe and Krause~\cite{stobbe2010efficient}) showed that a concave cardinality function can be decomposed as a sum of modular functions plus a combination of $|e|-1$ threshold potentials, our approximation technique will find a linear combination with $\lfloor |e|/2\rfloor$ threshold potentials. This amounts to the observation that any $k+1$ points $\{i,g(i)\}$ can be joined by $\lfloor k/2\rfloor + 1$  lines instead of using $k$. Overall though, the most significant advantage of \textsc{SparseCard} over existing reduction methods is its ability to find fast approximate solutions with sparse approximate reductions.}

\section{Experiments}
\label{sec:experiments}
In addition to its strong theoretical guarantees, \textsc{SparseCard} is very practical and leads to substantial improvements in benchmark image segmentation problems and hypergraph clustering tasks. 
We focus on image segmentation and localized hypergraph clustering tasks that simultaneously include component functions of large and small support, which are common in practice~\cite{shanu2016min,ene2017decomposable,veldt2020hypercuts,Liu-2021-localhyper,purkait2016clustering}. 
Image segmentation experiments were run on a laptop with a 2.2 GHz Intel Core i7 processor and 8GB of RAM. For our local hypergraph clustering experiments, in order to run a larger number of experiments we used a machine with 4 x 18-core, 3.10 GHz Intel Xeon gold processors with 1.5 TB RAM.
We consider public datasets previously made available for academic research, and use existing open source software for competing methods.\footnote{Image datasets:~\url{http://people.csail.mit.edu/stefje/code.html}. Hypergraph clustering datasets:~\url{www.cs.cornell.edu/~arb/data/}. DSFM algorithms: from~\url{github.com/lipan00123/DSFM-with-incidence-relations}; Hypergraph clustering algorithms:~\url{github.com/nveldt/HypergraphFlowClustering}.}

\subsection{Benchmark Image Segmentation Tasks}
\textsc{SparseCard} provides faster approximate solutions for standard image segmentation tasks previously used as benchmarks for DSFM~\cite{jegelka2013reflection,panli2018revisiting,ene2017decomposable}. We consider the \emph{smallplant} and \emph{octopus} segmentation tasks
from Jegelka et al.~\cite{jegelka2011submodularity,jegelka2013reflection}. These amount to minimizing a decomposable submodular function on a ground set of size $ |V| = 427 \cdot 640 = 273280$, where each $v \in V$ is a pixel from a $427 \times 640$ pixel image and there are three types of component functions.
The first type are unary potentials for each pixel/node, i.e., functions of support size 1 representing each node's bias to be in the output set. The second type are pairwise potentials from a 4-neighbor grid graph; pixels $i$ and $j$ share an edge if they are directly adjacent vertically or horizontally. The third type are region potentials of the form $\vf_e(A) =|A||e\backslash A|$ for $A \subseteq e$, where $e$ represents a superpixel region. The problem can be solved via maximum flow even without sophisticated reduction techniques for cardinality functions, as a regional potential function on $e$ can be modeled by placing a clique of edges on $e$. We compute an optimal solution using this reduction.

\paragraph{Comparison against exact reductions}
Compared with the exact reduction method, running \textsc{SparseCard} with $\varepsilon > 0$ leads to much sparser graphs, much faster runtimes, and a posteriori approximation factors that are significantly better than $(1+\varepsilon)$. 
For example, on the \emph{smallplant} dataset, when $\varepsilon = 1.0$, the returned solution is within a factor $1.004$ of optimality, even though the a priori guarantee is a $2$-approximation. The sparse reduced graph has only $0.013$ times the number of edges in the exact reduced graph, and solving the sparse flow problem is over two orders of magnitude faster than solving the problem on the exact reduced graph (which takes roughly 20 minutes on a laptop). 
In Table~\ref{tab:sparsecard} we list the sparsity, runtime, and a posteriori guarantee obtained for a range of $\varepsilon$ values on the \emph{smallplant} dataset using the superpixel segmentation with 500 regions.
\begin{table}[t!]
	\caption{
		Results from \textsc{SparseCard} for different $\varepsilon>0$ on the \emph{smallplant} instance with 500 superpixels.
		Sparsity is the fraction of edges in the approximate graph reduction compared with the exact reduction. Finding the exact solution on the dense exact reduced graph took $\approx$20 minutes.
	}
	\label{tab:sparsecard}
	\centering
	\begin{tabular}{llllllll}
			\toprule
$\varepsilon$ &1.0 &0.2336 &0.0546 &0.0127 &0.003 &0.0007 &0.0002  \\
Approx.$-1$ &$4\cdot 10^{-3}$ &$2\cdot 10^{-3}$ &$6\cdot 10^{-4}$&$6\cdot 10^{-5}$ &$3\cdot 10^{-5}$ &$7\cdot 10^{-6}$ &$7\cdot 10^{-7}$ \\
Sparsity &0.013 &0.017 &0.02 &0.035 &0.06 &0.108 &0.196  \\
Runtime &4.1 &5.6 &6.7 &11.5 &24.3 &41.4 &74.3 \\
			\bottomrule
	\end{tabular}
\vspace{-.8\baselineskip}
\end{table} 

\begin{figure}[t]
		\includegraphics[width=.265\linewidth]{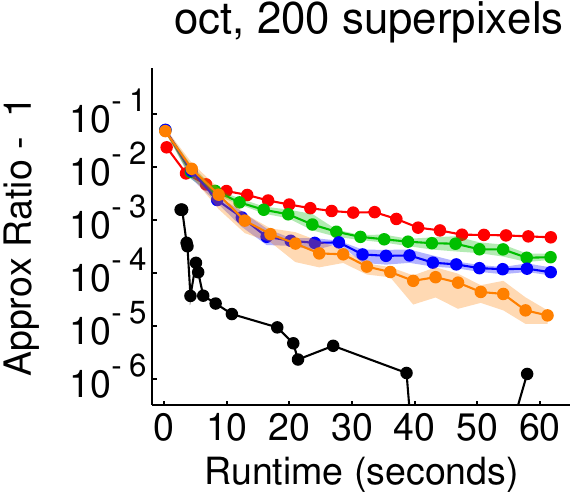} 
		\includegraphics[width=.235\linewidth]{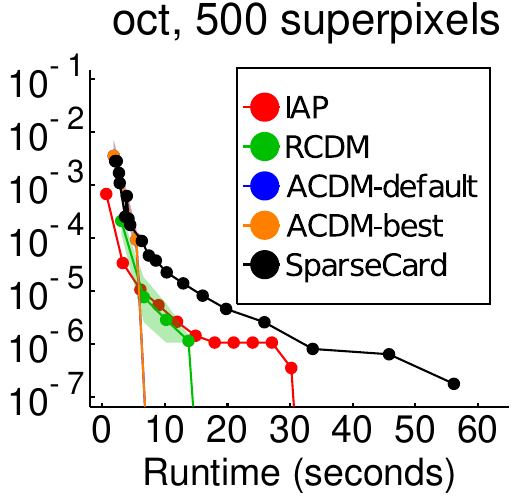} 
		\includegraphics[width=.235\linewidth]{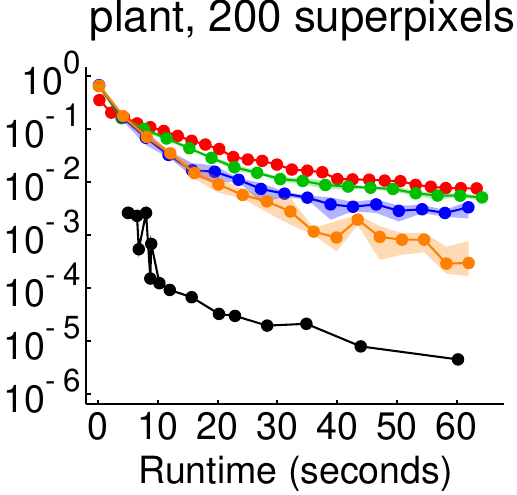} 
		\includegraphics[width=.235\linewidth]{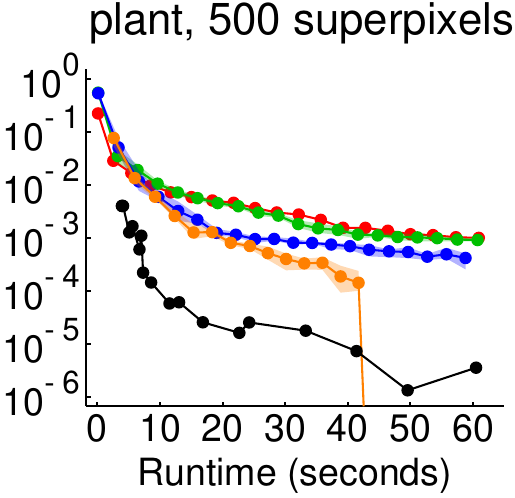}  
	\caption{Approximation factor minus 1 vs.\ runtime for solutions returned by \textsc{SparseCard} and competing methods on four image segmentation tasks. 
	\textsc{SparseCard} obtains faster approximate solutions on all but the easiest instance, where all methods obtain good results within a few seconds.
	We display the average of 5 runs for competing methods, with lighter colored region showing upper and lower bounds from these runs. \textsc{SparseCard} is deterministic and was run once for each $\varepsilon$ on a decreasing logarithmic scale. Our method maintains an advantage even against post-hoc best case parameters for competing approaches: ACDM-best is the best result obtained by running ACDM for a range of empirical parameters $c$ for each dataset and reporting the best result. The default is $c = 10$ (blue curve). Best post-hoc results for the plots from left to right were $c = 25, 10, 50, 25$. It is unclear how to determine the best $c$ in advance.
}
\vspace{-10pt}
\label{fig:images}
\end{figure}

\paragraph{Comparison against continuous optimization algorithms}
We also compare against recent C++ implementations of ACDM, RCDM, and Incidence Relation AP (an improved version of the standard AP method~\cite{nishihara2014convergence}) provided by Li and Milenkovic~\cite{panli2018revisiting}. These use the divide-and-conquer method of Jegelka et al.~\cite{jegelka2013reflection}, implemented specifically for concave cardinality functions, to solve the quadratic minimization oracle $\mathcal{O}_e$ for region potential functions. Although these continuous optimization methods come with no a priori approximation guarantees, we can compare them against \textsc{SparseCard} by computing a posteriori approximations obtained using intermediate solutions returned after every few hundred iterations.
In more detail, we extract the best level set $S_\lambda = \{i \colon x_i > \lambda\}$ from the vector of dual variables $\vx = (x_i)$ at various steps, and compute the approximation ratio $\argmin_\lambda f(S_\lambda)/f(S^*)$ where $S^*$ is the optimal solution determined via the exact max-flow solution. 
Figure~\ref{fig:images} displays approximation ratio versus runtime for four DSFM instances (two datasets $\times$ two superpixel segmentations). \textsc{SparseCard} was run for a range of $\varepsilon$ values on a decreasing logarithmic scale from $1$ to $10^{-4}$, and obtains significantly better results on all but the \emph{octopus} with 500 superpixels instance. This is the easiest instance; all methods obtain a solution within a factor $1.001$ of optimality within a few seconds. ACDM depends on a hyperparameter $c$ controlling the number of iterations in an outer loop. Even when we choose the best post-hoc $c$ value for each dataset, \textsc{SparseCard} maintains its overall advantage. Appendix~\ref{app:dsfm} provides additional details regarding the competing algorithms and their parameter settings.

We focus on comparisons with continuous optimization methods rather than other discrete optimization methods, as the former are better equipped for our goal of finding approximate solutions to DSFM problems involving functions of large support. To our knowledge, no implementations for the methods of Kolmogorov~\cite{kolmogorov2012minimizing} or Axiotis et al.~\cite{axiotis2021decomposable} exist. 
Meanwhile, IBFS~\cite{fix2013structured} is designed for finding exact solutions when all support sizes are small. Recent empirical results~\cite{ene2017decomposable}  confirm that this method is not designed to handle the large region potential functions we consider here.

\subsection{Hypergraph local clustering}
Graph reduction techniques have been frequently and successfully used as subroutines for hypergraph local clustering and semi-supervised learning methods~\cite{Liu-2021-localhyper,veldt2020hyperlocal,panli2017inhomogeneous,yin2017local}.
Replacing exact reductions with our approximate reductions can lead to significant runtime improvements without sacrificing on accuracy, and opens the door to running local clustering algorithms on problems where exact graph reduction would be infeasible. We illustrate this by using \textsc{SparseCard} as a subroutine for a method we previously designed called~\textsc{HyperLocal}~\cite{veldt2020hyperlocal}. This algorithm finds local clusters in a hypergraph by repeatedly solving hypergraph minimum $s$-$t$ cut problems, which could also be viewed as instances of Card-DSFM. \textsc{HyperLocal} was originally designed to handle only the \emph{$\delta$-linear} penalty $\vf_e(A) = \min  \{ |A|, |e \backslash A|, \delta\}$, for parameter $\delta \geq 1$, which can already be modeled sparsely with a single CB-gadget. \textsc{SparseCard} makes it possible to sparsely model any concave cardinality penalty. We specifically use approximate reductions for the weighted clique penalty $\vf_e(A) = (|e|-1)^{-1}|A||e \backslash A|$, 
the square root penalty $\vf_e(A) = \sqrt{\min\{|A|, |e \backslash A|\}}$, 
and the sublinear power function penalty $\vf_e(A) = (\min\{|A|, |e \backslash A|\})^{0.9}$,
all of which require $O(|e|^2)$ edges to model exactly using previous reduction techniques. Weighted clique penalties in particular have been used extensively in hypergraph clustering~\cite{Agarwal2006holearning,yin2017local,kumar2020modularity,zien1999}, including by methods specifically designed for local clustering and semi-supervised learning~\cite{jianbo2018,yin2017local,Zhou2006learning}. 

\paragraph{Stackoverflow hypergraph}
We consider a hypergraph clustering problem where nodes are 15.2M questions on \verb|stackoverflow.com| and each of the 1.1M hyperedges defines a set of questions answered by the same user. The mean hyperedge size is 23.7, the maximum size is over 60k, and there are 2165 hyperedges with at least 1000 nodes.  Questions with the same topic tag (e.g., ``common-lisp'') constitute small labeled clusters in the dataset. We previously showed that \textsc{HyperLocal} can detect clusters quickly with the $\delta$-linear penalty by solving localized $s$-$t$ cut problems near a seed set. Applying exact graph reductions for other concave cut penalties is infeasible, due to the extremely large hyperedge sizes, and we found that using a clique expansion after simply removing large hyperedges performed poorly~\cite{veldt2020hyperlocal}. Using \textsc{SparseCard} as a subroutine opens up new possibilities.

\paragraph{Experimental setup and parameter settings}
We follow our previous experimental setup~\cite{veldt2020hyperlocal} for detecting localized clusters in the Stackoverflow hypergraph. We focus on a set of 45 local clusters, all of which are question topics involving between 2,000 and 10,000 nodes.
For each cluster, we generate a random seed set by selecting 5\% of the nodes in the target cluster uniformly at random, and then add neighboring nodes to the seed set to grow it into a larger input set $Z$ to use for~\textsc{HyperLocal} (see~\cite{veldt2020hyperlocal} for details). We set $\delta = 5000$ for the $\delta$-linear hyperedge cut function and set a locality parameter for HyperLocal to be $1.0$ for all experiments. With this setup, using~\textsc{HyperLocal} with the $\delta$-linear penalty will then reproduce our original experimental setup~\cite{veldt2020hyperlocal}. Our goal is to show how using \textsc{SparseCard} leads to fast and often improved results for alternative penalties that could not previously been used.

\begin{figure}[t]
	\centering
	\includegraphics[width = .9\linewidth]{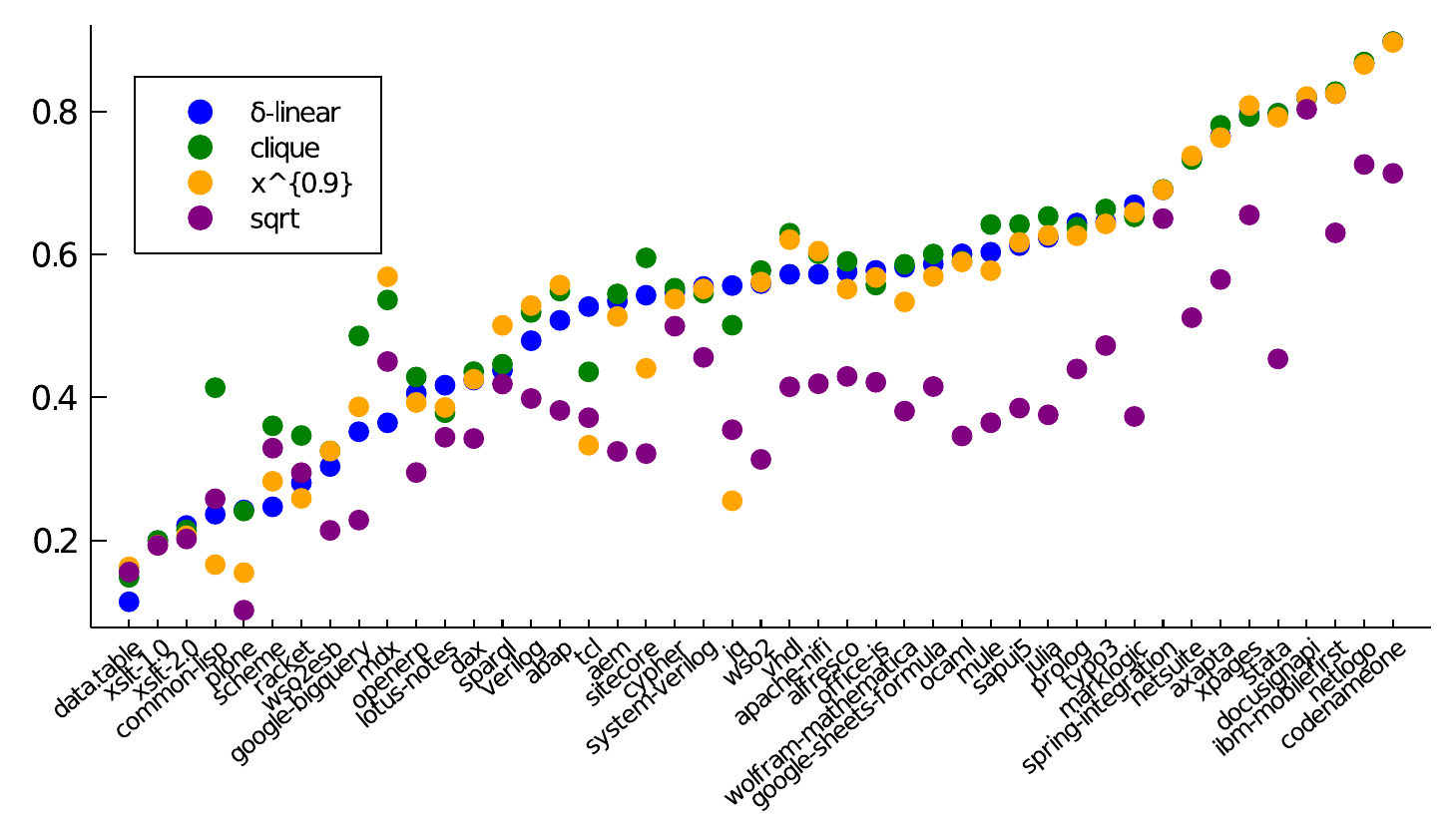}
	\caption{Average F1 score obtained for each localized clustering using 4 different hyperedge cut penalties. For clique, $x^{0.9}$, and the square root penalties, we used an approximate reduction with $\varepsilon = 1.0$. The clique penalty had the highest average F1 score on 26 clusters, the $\delta$-linear had the highest on 10 clusters, and $x^{0.9}$ had the highest average score on the remaining 9 clusters.}
	\label{fig:meanplots}
\end{figure}

\paragraph{Experimental results}
In Figure~\ref{fig:meanplots} we show the mean F1 score for each cluster (across the ten random seed sets) obtained by the $\delta$-linear penalty and the three alternative penalties when $\varepsilon = 1.0$. The clique penalty obtained the highest average F1 score on 26 clusters, the $\delta$-linear obtained the highest average score on 10 clusters, and the sublinear penalty obtained the best average score on the remaining 9 clusters. Importantly, cut penalties that previously could not be used on this dataset (\emph{clique}, \emph{$x^{0.9}$}) obtain the best results for most clusters. The square root penalty does not perform particularly well on this dataset, but it is instructive to consider its runtime (Figure~\ref{sqrtrun}). Theorem~\ref{thm:tight} shows that asymptotically this function has a worst-case behavior in terms of the number of CB-gadgets needed to approximate it. We nevertheless obtain reasonably fast results for this penalty function, indicating that our techniques can provide effective sparse reductions for any concave cardinality function of interest. 

\begin{figure}[t!]
	\begin{subfigure}[t]{0.495\linewidth}
		\centering
		\includegraphics[width=\linewidth]{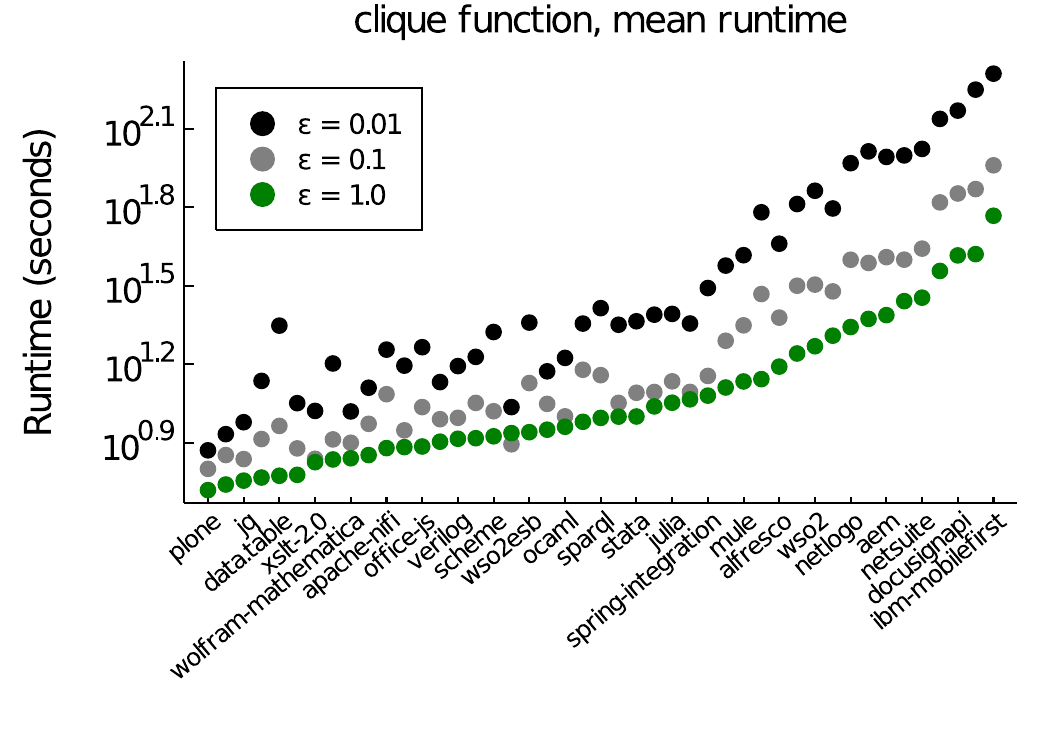}
		\caption{Clique Runtime}
	\end{subfigure}
	\begin{subfigure}[t]{0.495\linewidth}
	\centering
	\includegraphics[width=\linewidth]{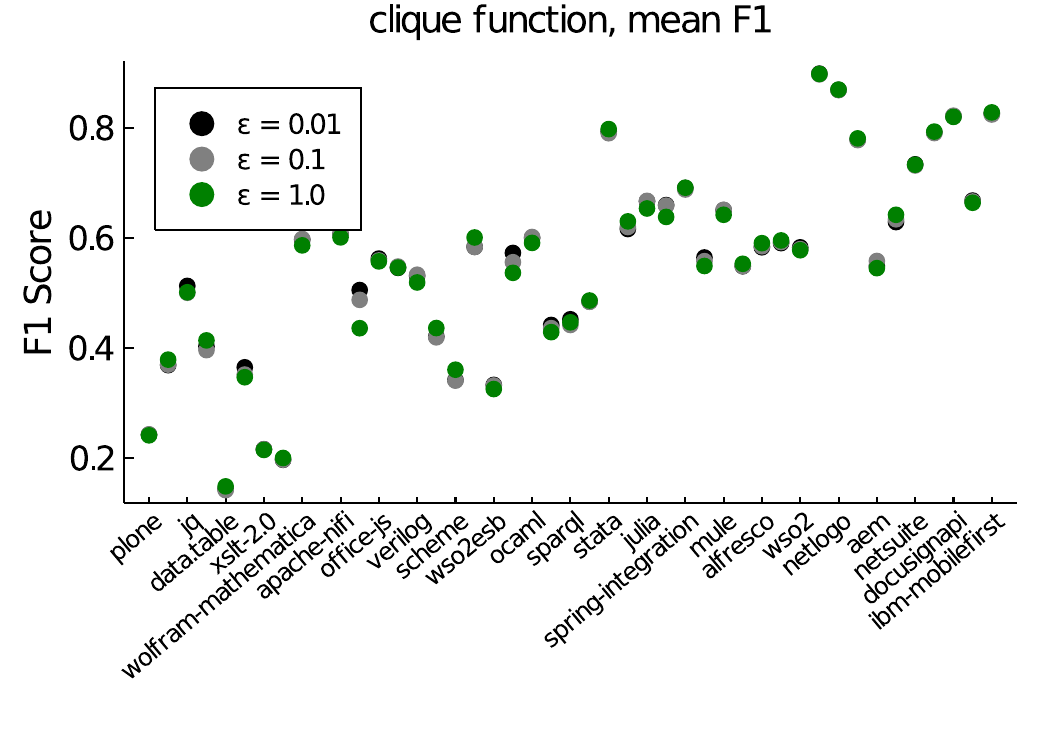}
	\caption{Clique F1 Scores}
\end{subfigure}
	\begin{subfigure}[t]{0.495\linewidth}
	\centering
	\includegraphics[width=\linewidth]{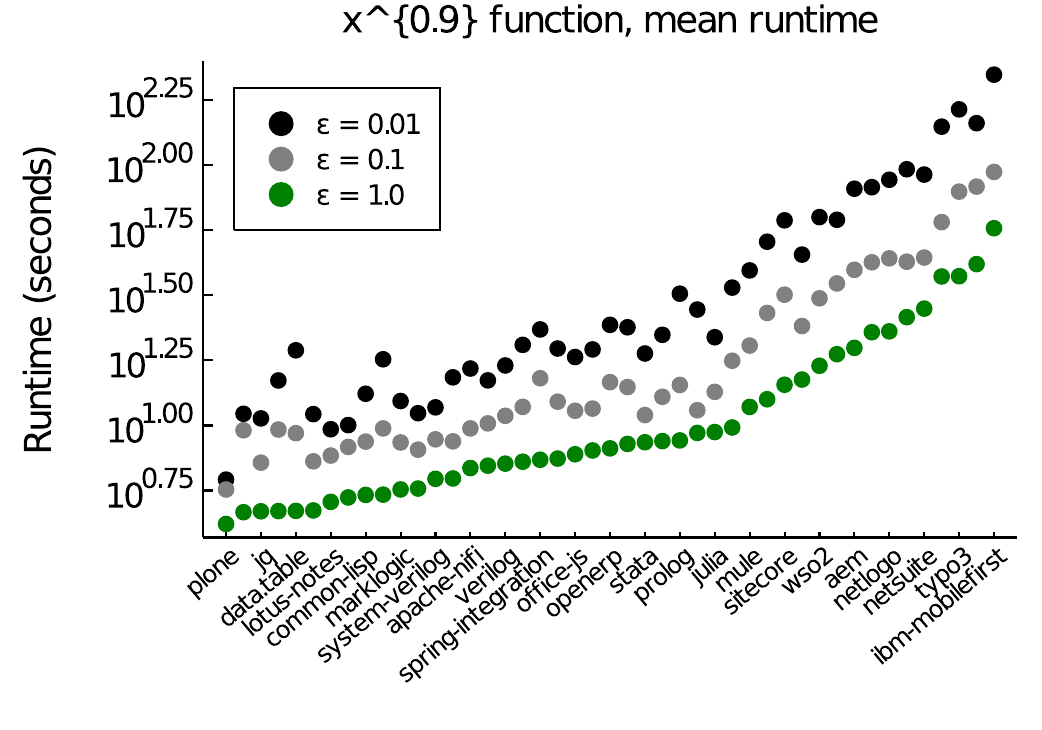}
	\caption{$x^{0.9}$ Runtime}
\end{subfigure}
\begin{subfigure}[t]{0.495\linewidth}
	\centering
	\includegraphics[width=\linewidth]{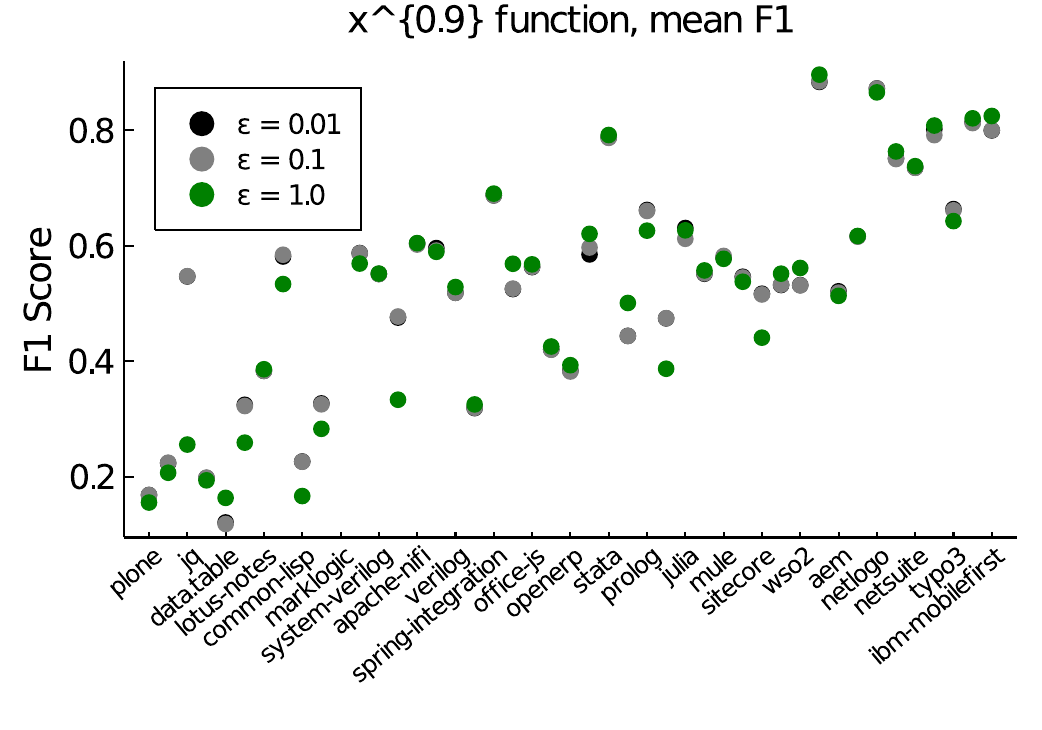}
	\caption{$x^{0.9}$ F1 Scores}
\end{subfigure}
	\begin{subfigure}[t]{0.495\linewidth}
	\centering
	\includegraphics[width=\linewidth]{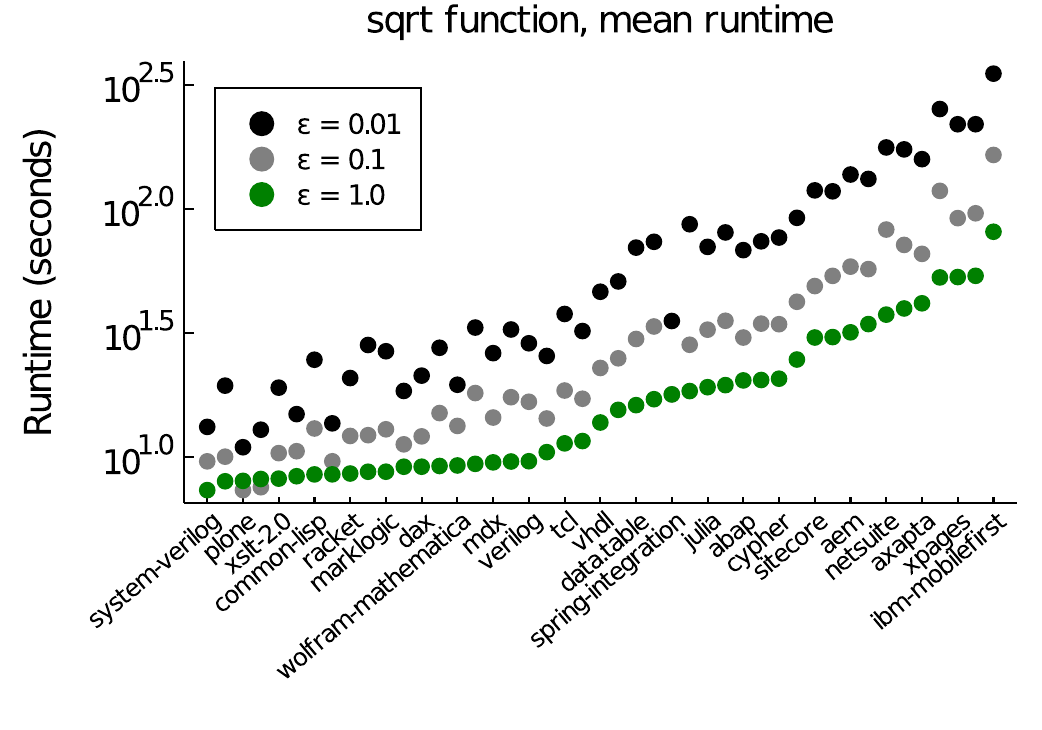}
	\caption{Square Root Runtime}
\end{subfigure}
\begin{subfigure}[t]{0.495\linewidth}
	\centering
	\includegraphics[width=\linewidth]{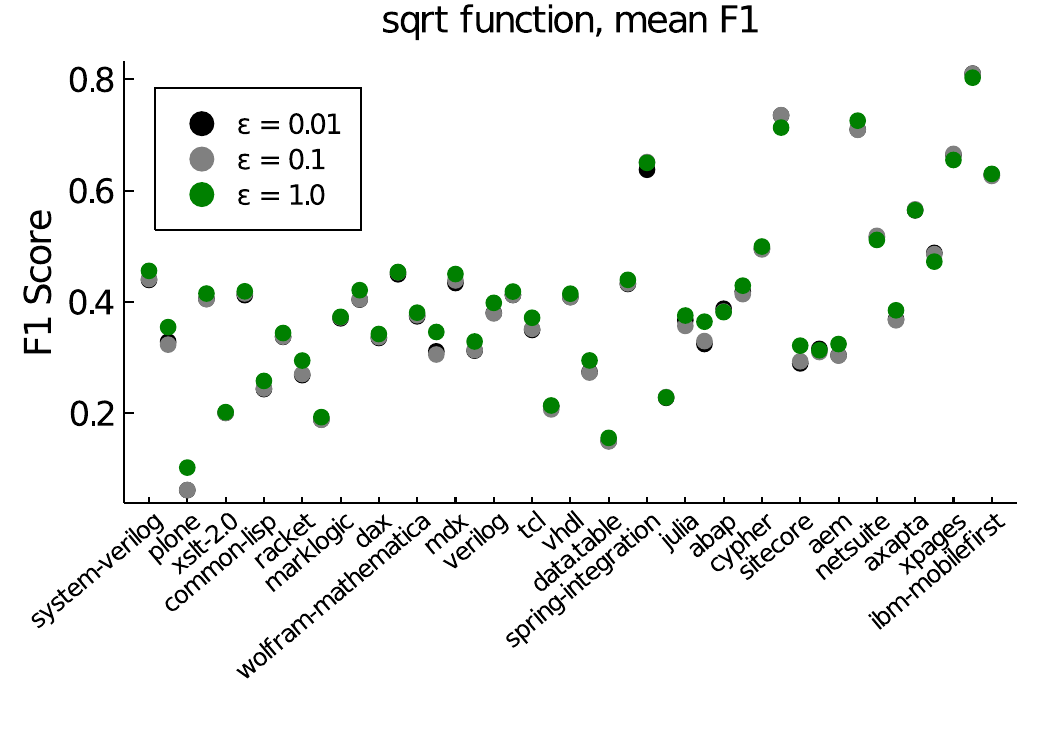}
	\caption{Square Root F1 Scores}
	\label{sqrtrun}
\end{subfigure}
	\caption{We display differences in runtimes and F1 scores when using different $\varepsilon$ values when approximating three hyperedge cut penalties. Using a larger $\varepsilon$ provides a significant runtime savings with virtually no affect on F1 scores. Clusters have been re-arranged in the horizontal axis for each hyperedge cut penalty for easier visualization.}
	\label{fig:runtimes}
\end{figure}
We also ran experiments with $\varepsilon = 0.1$ and $\varepsilon = 0.01$, which led to noticeable increases in runtime but only very slight changes in F1 scores (Figure~\ref{fig:runtimes}). This indicates why exact reductions are not possible in general, while also showing that our sparse approximate reductions serve as fast and very good proxies for exact reductions.

\section{Conclusion and Discussion}
\label{sec:conclusion}
We have introduced the notion of an augmented cut sparsifier, which approximates a generalized hypergraph cut function with a sparse directed graph on an augmented node set. Our approach relies on a connection we highlight between graph reduction strategies and piecewise linear approximations to concave functions. Our framework leads to more efficient techniques for approximating hypergraph $s$-$t$ cut problems via graph reduction, improved sparsifiers for co-occurrence graphs, and fast algorithms for approximately minimizing cardinality-based decomposable submodular functions. 

As noted in Section~\ref{sec:present}, an interesting open question is to establish and study analogous notions of augmented \emph{spectral} sparsification, given that spectral sparsifiers provide a useful generalization of cut sparsifiers in graphs~\cite{spielmanteng2011}. 
One way to define such a notion is to apply existing definitions of submodular hypergraph Laplacians~\cite{panli_submodular,yoshida2019cheeger} to both the original hypergraph and its sparsifier. This requires viewing our augmented sparsifier as a hypergraph with splitting functions of the form $\vw_e(A) = a \cdot\min \{|A|, |e \backslash A|,b \}$, corresponding to hyperedges with cut properties that can be modeled by a cardinality-based gadget. From this perspective, augmented spectral sparsification means approximating a generalized hypergraph cut function with another hypergraph cut function involving simplified splitting functions. While this provides one possible definition for augmented spectral sparsification, it is not clear whether the techniques we have developed can be used to satisfy this definition. Furthermore, it is not clear whether obtaining such a sparsifier would imply any immediate runtime benefits for approximating the spectra of generalized hypergraph Laplacians, or for solving generalized Laplacian systems~\cite{fujii2018polynomialtime,panli2020qdsfm}. We leave these as questions for future work.

While our work provides the optimal reduction strategy in terms of cardinality-based gadgets, this is more restrictive than optimizing over all possible gadgets for approximately modeling hyperedge cut penalties. Optimizing over a broader space of gadgets poses another interesting direction for future work, but is more challenging in several ways. First of all, it is unclear how to even define an optimal reduction when optimizing over arbitrary gadgets, since it is preferable to avoid both adding new nodes and adding new edges, but the tradeoff between these two goals is not clear. Another challenge is that the best reduction may depend heavily on the splitting function we wish to reduce, which makes developing a general approach difficult. A natural next step would be to at least better understand lower bounds on the number of edges and auxiliary nodes needed to model different cardinality-based splitting functions. While we do not have any concrete results, there are several indications that cardinality-based gadgets may be nearly optimal in many settings. For example, star expansions and clique expansions provide a more efficient way to model linear and quadratic splitting functions respectively, but modeling these functions with cardinality-based gadgets only increases the number of edges by roughly a factor two.

Finally, we find it interesting that using auxiliary nodes and directed edges makes it possible to sparsify the complete graph using only $O(n \varepsilon^{-1/2} \log \log \frac{1}{\varepsilon})$ edges, whereas standard sparsifiers require $O(n \varepsilon^{-2})$. We would like to better understand whether both directed edges \emph{and} auxiliary nodes are necessary for making this possible, or whether improved approximations are possible using only one or the other.

\appendix
\section{Proofs of Lemmas in Section~\ref{sec:reductions}}
\label{app:proofs}
\paragraph{Proof of Lemma~\ref{lem:ccb2pl}}
\begin{lemmaintro}
	Let $\hat{\vf}$ be the continuous extension for a function $\vw \in \mathcal{S}_r$, shown in~\eqref{cextt}. This function is in the class $\mathcal{F}_r$, and has exactly $J$ positive-sloped linear pieces, and one linear piece of slope zero.
\end{lemmaintro}
\begin{proof}
	Define $b_0 = 0$ for notational convenience. The first three conditions in Definition~\ref{def:plfun} can be seen by inspection, recalling that $0< a_j$ and $0< b_j \leq r$ for all $j \in [J]$. Observe that $\hat{\vf}$ is linear over the interval $[b_{i-1}, b_{i})$ for $i \in [J]$, since for $x \in [b_{i-1}, b_{i})$,
	\begin{align*}
	\hat{\vf}(x) &= \sum_{j = 1}^J a_j \cdot \min \{ x, b_j \} = \sum_{j = 1}^{i-1} a_j  b_j + x \cdot \sum_{j = i}^J a_j \,.
	\end{align*}
	In other words, the $i$th linear piece of $\hat{\vf}$, defined over $x \in [b_{i-1}, b_{i})$ is given by
	$\hat{\vf}^{(i)}(x) = I_i + S_i x,$
	where the intercept and slope terms are given by
$	I_i = \sum_{j=1}^{i-1} a_j  b_j$ and  $S_i =\sum_{j = i}^J a_j$.
	For the first $J$ intervals of the form $[b_{i-1},b_i)$, the slopes are always positive but strictly decreasing. Thus, there are exactly $J$ positive sloped linear pieces. The final linear piece is a flat line, since $\hat{\vf}(x) = \sum_{j = 1}^J a_jb_j$ for all $x \geq b_J$.
	The concavity of $\hat{\vf}$ follows directly from the fact that it is a continuous and piecewise linear function with decreasing slopes.
\end{proof}

\paragraph{Proof of Lemma~\ref{lem:pl2ccb}}
\begin{lemmaintro}
	Let $\vf$ be a function in $\mathcal{F}_r$ with $J+1$ linear pieces. Let $b_i$ denote the $i$th breakpoint of $\vf$, and $m_i$ denote the slope of the $i$th linear piece of $\vf$. Define vectors $\va, \vb \in \mathbb{R}^J$ where $\vb(i) = b_i$ and $\va(i) = a_i = m_i - m_{i+1}$ for $i \in [J]$. If ${\vw}$ is the $r$-CCB function parameterized by vectors $(\va, \vb)$, then $\vf$ is the continuous extension of ${\vw}$.
\end{lemmaintro} 
\begin{proof}
	 Since $\vf$ is in $\mathcal{F}_r$, it has $J$ positive-sloped linear pieces and one flat linear piece, and therefore it has exactly $J$ breakpoints: $0 < b_1 < b_2 < \hdots < b_J$.  Let $\vb = (b_j)$ be the vector storing these breakpoints. For convenience we define $b_0 = 0$, though $b_0$ is not stored in $\vb$. By definition, $\vf$ is constant for all $x \geq r$, which implies that $b_J \leq r$. 
	
	Let $f_i = \vf(b_i)$. For $i \in [J]$, the positive slope of the $i$th linear piece of $\vf$, which occurs in the range $[b_{i-1}, b_i]$, is given by
	\begin{equation}
	m_i = \frac{f_i - f_{i-1}}{b_{i} - b_{i-1}}.
	\end{equation}
	The $i$th linear piece of $\vf$ is given by
	\begin{equation}
	\vf^{(i)}(x) = m_i (x- b_{i-1}) + f_{i-1} \hspace{.5cm} \text{ for $x \in [b_{i-1},b_i]$}.
	\end{equation}
	The last linear piece of $\vf$ is a flat line over the interval $x \in [b_J, \infty)$, i.e., $m_{J+1} = 0$.
	Since $\vf$ has positive and strictly decreasing slopes, we can see that $a_i = m_i - m_{i+1} > 0$ for all $i \in [J]$.
	
	Let $\vw$ be the order-$J$ CCB function constructed from vectors $(\va, \vb)$, and let $\hat{\vf}$ be its resulting continuous extension:
	\begin{equation}
	\hat{\vf} = \sum_{j = 1}^J a_j \cdot \min \{x, b_j\}.
	\end{equation}
	We must check that $\hat{\vf} = \vf$. By Lemma~\ref{lem:ccb2pl}, we know that $\hat{\vf}$ is in $\mathcal{F}_r$ and has exactly $J+1$ linear pieces. The functions will be the same, therefore, if they share the same values at breakpoints. Evaluating $\hat{\vf}$ at an arbitrary breakpoint $b_i$ gives:
	\begin{equation}
	\hat{\vf}(b_i) = \left( \sum_{j = 1}^{i-1} a_j \cdot b_j \right) + b_i \cdot  \left( \sum_{j = i}^J a_j \right) = \left( \sum_{j = 1}^{i-1} a_j \cdot b_j \right) + b_i \cdot m_i .
	\end{equation}
	We first confirm that the functions coincide at the first breakpoint:
	\begin{align*}
	\hat{\vf}(b_1) 
	&= b_1 \cdot m_1 = b_1 \cdot \frac{f_1 - f_0}{b_1 - b_0} = b_1 \frac{f_1}{b_1} = f_1.
	\end{align*}
	For any fixed $i \in \{2, 3, \hdots , J\}$,
	\begin{align*}
	\hat{\vf}(b_i) - \hat{\vf}(b_{i-1}) &= \left( \sum_{j = 1}^{i-1} a_j   b_j \right) + b_i   m_i  - 
	\left( \sum_{j = 1}^{i-2} a_j   b_j \right) - b_{i-1}   m_{i-1}\\
	&= a_{i-1}   b_{i-1} + b_i   m_i - b_{i-1}   m_{i-1}\\
	&= (m_{i-1} - m_{i})   b_{i-1} +  b_i   m_i - b_{i-1}   m_{i-1}\\
	&= m_i (b_i - b_{i-1}) = f_{i} - f_{i-1}.
	\end{align*}
	Since $\vf(b_1) = \hat{\vf}(b_1)$ and $\vf(b_i) - \vf(b_{i-1}) = \hat{\vf}(b_i) - \hat{\vf}(b_{i-1})$ for $i \in \{2, 3, \hdots , t\}$, we have $\vf(b_i) = \hat{\vf}(b_i)$ for $i \in [J]$. Therefore, $\vf$ and $\hat{\vf}$ are the same piecewise linear function.
\end{proof}
\section{Sparsification for Generalized Splitting Functions}
\label{app:general}
In Sections~\ref{sec:spagap} and~\ref{sec:reductions} we focused on sparsification techniques for representing splitting functions that are symmetric and penalize only \emph{cut} hyperedges:
\begin{align*}
\vw_e(S) &= \vw(e \backslash S)  \text{ for all $S \subseteq e$}\\
\vw_e(e) &= \vw_e(\emptyset) = 0.
\end{align*}
These assumptions are standard for generalized hypergraph cut problems~\cite{veldt2020hypercuts, panli2017inhomogeneous, panli_submodular}, and lead to the clearest exposition of our main results. In this appendix, we extend our sparse approximation techniques so that they apply even if we remove these restrictions. This will allow us to obtain improved techniques for approximately solving a certain class of decomposable submodular functions (see Section~\ref{sec:runtime}). Formally, our goal is to minimize
\begin{equation}
\label{soscard}
\minimize_{S \subseteq V} f(S) = \sum_{e \in \mathcal{E}} \vw_e(S \cap e) ,
\end{equation}
where each $\vw_e$ is a submodular cardinality-based function, that is not necessarily symmetric and does not need to equal zero when the hyperedge $e$ is uncut.
Our proof strategy for reducing this more general problem to a graph $s$-$t$ cut problem closely follows the same basic set of steps used in Section~\ref{sec:reductions} for the special case.

\subsection{Submodularity Constraints for Cardinality-Based Functions}
We first provide a convenient characterization of general cardinality-based submodular functions. By \emph{general} we mean the splitting function does not need to be symmetric nor does it need to have a zero penalty when the hyperedge is uncut.
\begin{lemma}
	\label{lem:gencard}
	Let $\vw_e$ be a general submodular cardinality-based splitting function on a $k$-node hyperedge $e$, and let $w_i$ denote the penalty for any $A \subseteq e$ with $|A| = i$. Then for $i \in  \{1, 2, \hdots ,k-1\}$
	\begin{equation}
	\label{genconstraint}
	2w_i \geq w_{i-1} + w_{i+1}.
	\end{equation}
\end{lemma}
\begin{proof}
Let $v_1, v_2, \hdots, v_k$ denote the nodes in the hyperedge. Submodularity means that for all $A, B \subseteq e$, $\vw(A) + \vw(B) \geq \vw(A\cup B) + \vw(A \cap B)$. In order to show inequality~\eqref{genconstraint}, simply set $A = \{v_1, v_2, \hdots, v_i\}$ and $B = \{v_2, v_2, \hdots , v_{i+1}\}$ and the result follows.
\end{proof}
To simplify our analysis, as we did for the symmetric case, we will define a set of functions that is virtually identical to these splitting functions on $k$-node hyperedges, but are defined over integers from $0$ to $k$ rather than on subsets of a hyperedge.
\begin{definition}
	A $k$-GSCB (Generalized Submodular Cardinality-Based) integer function is a function $\vw : \{0\} \cup [k] \rightarrow \mathbb{R}^+$ satisfying $2\vw(i) \geq \vw(i-1) + \vw(i+1)$ for all $i \in [k-1]$.
\end{definition}

\subsection{Combining Gadgets for Generalized SCB Functions}
Our goal is to show how to approximate $k$-GSCB integer functions using piecewise linear functions with few linear pieces. This in turn corresponds to approximating a hyperedge splitting function with a sparse gadget. In order for this to work for our more general class of splitting functions, we use a slight generalization of an asymmetric gadget we introduced in previous work~\cite{veldt2020hypercuts}.
\begin{definition}
	The asymmetric cardinality-based gadget (ACB-gadget) for a $k$-node hyperedge $e$ is parameterized by scalars $a$ and $b$ and constructed as follows:
	\begin{itemize}
		\item Introduce an auxiliary vertex $v_e$.
		\item For each $v \in e$, introduce a directed edge from $v$ to $v_e$ with weight $a \cdot (k-b)$, and a directed edge from $v_e$ to $v$ with weight $a \cdot b$.
		\end{itemize}
\end{definition}
The ACB-gadget models the following $k$-GSCB integer function:
\begin{equation}
\label{acbfun}
\vw_{a,b}(i) = a \cdot \min \{ i \cdot (k-b), (k-i) \cdot b \}.
\end{equation}
To see why, consider where we must place the auxiliary node $v_e$ when solving a minimum $s$-$t$ cut problem involving the ACB-gadget. If we place $i$ nodes on the $s$-side, then placing $v_e$ on the $s$-side has a cut penalty of $ab(k-i)$, whereas placing $v_e
$ on the $t$-side gives a penalty of $ai (k-b)$. To minimize the cut, we choose the smaller of the two options.

Previously we showed that asymmetric splitting functions can be modeled exactly by a combination of $k-1$ ACB-gadgets~\cite{veldt2020hypercuts}. As we did in Section~\ref{sec:reductions} for symmetric splitting functions, we will show here that a much smaller number of gadgets suffices if we are content to approximate the cut penalties of an asymmetric splitting function. In our previous work we enforced the constraint $\vw_e(\emptyset) = \vw_e(0) = 0$ even for asymmetric splitting functions, but we remove this constraint here. In order to model the cut properties of an arbitrary GSCB splitting function, we define a combined gadget involving multiple ACB-gadgets, as well as edges from each node $v \in e$ to the source and sink nodes of the graph. The augmented cut function for the resulting directed graph $\hat{G} = (V \cup \mathcal{A} \cup \{s,t\}, \hat{E})$ will then be given by $\cut_{\hat{G}}(S) = \min_{T \subseteq \mathcal{A}} \dircut_{\hat{G}} (\{s\} \cup S \cup T)$ for a set $S\subseteq V$, where $\dircut$ is the directed cut function on $\hat{G}$. Finding a minimum $s$-$t$ cut in $\hat{G}$ will solve objective~\eqref{soscard}, or equivalently, the cardinality-based decomposable submodular function minimization problem.
	
\begin{definition}
	\label{def:kcg}
	A $k$-CG function ($k$-node, combined gadget function) $\hat{\vw}$ of order $J$ is a $k$-GSCB integer function that is parameterized by scalars $z_0$, $z_k$, and $(a_j, b_j)$ for $j \in [J]$. The function has the form:
	\begin{equation}
	\label{eq:kcg}
	\hat{\vw}(i) = z_0 \cdot (k-i) + z_k\cdot i + \sum_{j=1}^J a_j \min \{ i \cdot (k-b_j), (k-i) \cdot b_j \}.
	\end{equation} 
	The scalars parameterizing $\hat{\vw}$ satisfy
	\begin{align*}
	&b_j > 0, a_j > 0 \text{ for all $j \in [J]$}\\
	&b_j < b_{j+1} \text{ for all $j \in [J-1]$}\\
	&b_J < k\\
	&z_0 \geq 0 \text{ and } z_k \geq 0.
	\end{align*}
\end{definition}
Conceptually, the function shown in~\eqref{eq:kcg} represents a combination of $J$ ACB-gadgets for a hyperedge $e$, where additionally for each node $v \in e$ we have place a directed edge from a source node $s$ to $v$ of weight $z_0$, and an edge from $v$ to a sink node $t$ with weight $z_k$. 

The continuous extension of the $k$-CG function~\eqref{eq:kcg} is defined to be:
\begin{equation}
\label{eq:kgcext}
\hat{\vf}(x) = z_0 \cdot (k-x) + z_k\cdot x + \sum_{j=1}^J a_j \min \{ x \cdot (k-b_j), (k-x) \cdot b_j \} \text{ for $x \in [0,k]$.}
\end{equation}

\begin{lemma}
	\label{lem:context}
	The continuous extension $\hat{\vf}$ of $\hat{\vw}$ is nonnegative over the interval $[0,k]$, piecewise linear, concave, and has exactly $J+1$ linear pieces.
\end{lemma}
\begin{proof}
	Nonnegativity follows quickly from the positivity of $z_0$, $z_k$, and $(a_i, b_i)$ for $i \in [J]$, and $b_J < k$. For other properties, we begin by re-writing the function as
	\begin{align}
	\hat{\vf}(x) &= z_0 \cdot (k-x) + z_k\cdot x + \sum_{j=1}^J a_j \min \{ x \cdot (k-b_j), (k-x) \cdot b_j \}\\
	&= k z_0 + x( z_k - z_0)+ k \cdot \sum_{j=1}^J a_j \min \{ x, b_j \}  - x \cdot \sum_{j =1}^J a_j b_j\\
	\label{eq:hatflast}
	&= k z_0 + x( z_k - z_0)+ kx \cdot \sum_{j: x < b_j } a_j + k \cdot \sum_{j: x \geq b_j} a_j b_j  - x \cdot \sum_{j =1}^J a_j b_j.
	\end{align}
Define
\begin{align*}
\beta = \sum_{j = 1}^J a_j b_j, \hspace{1cm} \beta_t = \sum_{j = 1}^t a_j b_j, \hspace{1cm} \alpha_t = \sum_{j = t+1}^J a_j,
\end{align*}
and observe that $\beta_t$ is strictly increasing with $t$, and $\alpha_t$ is strictly \emph{decreasing} with $t$. Define $b_0 = 0$ and $b_{J +1} = k$ for notational convenience. For any $t \in \{0\} \cup [J]$, the function is linear over the interval $[b_{t}, b_{t+1})$, since for $x \in [b_t, b_{t+1})$, we have
\begin{align*}
\hat{\vf}(x) &= k z_0 + x( z_k - z_0)+ kx \cdot \sum_{j: x < b_j } a_j + k \cdot \sum_{j: x \geq b_j} a_j b_j  - x \cdot \sum_{j =1}^J a_j b_j \\
&= k z_0 + x( z_k - z_0)+ kx \sum_{j = t+1}^J a_j + k \cdot \sum_{j = 1}^t a_j b_j - x \sum_{j = 1}^J a_j b_j\\
&= k z_0 + x( z_k - z_0)+ kx \alpha_t + k \beta_t - x \beta.
\end{align*}
Thus, $\hat{\vf}$ is piecewise linear. Furthermore, the slope of the line over the interval $[b_t, b_{t+1})$ is $(z_k - z_0 - \beta  + k \alpha_t )$, which strictly decreases as $t$ increases. The fact that all slopes are distinct means that there are exactly $J+1$ linear pieces, and the fact that these slopes are decreasing means that the function is concave over the interval $[0,k]$.
\end{proof}

\begin{lemma}
	\label{lem:reversecontext}
	For every function $\vf$ that is nonnegative, piecewise linear with $J+1$ linear pieces, and concave over the interval $[0,k]$, there exists some $k$-CG function $\vw$ of order $J$ such that $\vf$ is the continuous extension of $\vw$.
\end{lemma}
\begin{proof}
	The function $\vw$ will be defined by choosing parameters $z_0$, $z_k$, and $(a_j, b_j)$ for $j \in [J]$. Let $\hat{\vf}$ denote the continuous extension of the function $\vw$ that we will build. From the proof of Lemma~\ref{lem:context}, we know that the parameter $b_j$ will correspond to the $j$th breakpoint of $\hat{\vf}$. Therefore, given $\vf$, we set $b_j$ to be the $j$th breakpoint of the function $\vf$, so that the functions match at breakpoints. For convenience, we also set $b_0 = 0$ and $b_{J+1}=k$. We then set $z_0 = \vf(0)/k$ and $z_k = \vf(k)/k$, to guarantee that $\hat{\vf}(0) = \vf(0)$ and $\hat{\vf}(k) = \vf(k)$. In order to set the $a_j$ values, we first compute the slopes of each line of $\vf$. Let $f_j = \vf(b_j)$ for $j \in \{0\}\cup [J+1]$. The $j$th linear piece of $\vf$ has the slope:
	\begin{equation*}
	m_i = \frac{f_i- f_{i-1}}{b_i - b_{i-1}}.
	\end{equation*}
	Finally, for $j \in [J]$ we set $a_j = \frac{1}{k}(m_{j} - m_{j+1})$. All of our chosen parameters satisfy the conditions of Definition~\ref{def:kcg}, so it simply remains to check that $\vf$ and $\hat{\vf}$ coincide at breakpoints.
	
	Let $t \in [J]$. Using~\eqref{eq:hatflast} to evaluating $\hat{\vf}$ at breakpoint $b_t$, we get
	\begin{align}
	\label{fhat}
	\hat{\vf}(b_t) &= f_0 + \frac{b_t}{k} ( f_k - f_0) + k b_t \sum_{j = t+1}^J a_j + k \sum_{j=1}^t a_j b_j - b_t \sum_{j=1}^J a_j b_j.
	\end{align}
	We can simplify several terms using the fact that $a_j = \frac{1}{k}(m_{j} - m_{j+1})$. First of all,
	\begin{align*}
	k \sum_{j = t+1}^J a_j &= \sum_{j = t+1}^J [m_{j} - m_{j+1}] = m_{t+1} - m_{J+1}.
		\end{align*}
		Furthermore,
	\begin{align*}
	k \sum_{j=1}^t a_j b_j &=  \sum_{j = 1}^t (m_{j} - m_{j+1})b_j = m_1 b_1 - m_{t+1} b_t  + \sum_{j = 2}^{t} m_j (b_{j} - b_{j-1}) \\
									&=(f_1 - f_0) - m_{t+1} b_t + \sum_{j = 2}^{t} [f_j - f_{j-1}] = (f_1 - f_0) - m_{t+1} b_t + f_t - f_{1}\\
									&= f_t - f_0 - m_{t+1} b_t.
	\end{align*}
	Similarly, we see that $\sum_{j = 1}^J a_j b_j = \frac{1}{k} \left( f_J - f_0 - m_{J+1} b_J \right)$. Plugging this into~\eqref{fhat}, we get
	\begin{align*}
	\label{fhatnew}
	\hat{\vf}(b_t) &= f_0 + \frac{b_t}{k} ( f_{J+1} - f_0)  + b_t(m_{t+1} - m_{J+1}) + f_t - f_0 - m_{t+1} b_t - \frac{b_t}{k} \left( f_J - f_0 - m_{J+1} b_J \right)\\
	&= \frac{b_t}{k} f_{J+1}  - b_tm_{J+1}+ f_t  - \frac{b_t}{k} \left( f_J - m_{J+1} b_J \right)\\
	&= f_t + \frac{b_t}{k} (f_{J+1}  - f_J) - b_t m_{J+1} \left(1- \frac{b_J}{k} \right)\\
	&= f_t + \frac{b_t}{k} (f_{J+1}  - f_J) - b_t \left(\frac{f_{J+1} - f_J}{k -b_J}\right)\left(\frac{k - b_J}{k} \right)= f_t = \vf(b_t).
	\end{align*}
So we see that $\vf = \hat{\vf}$ at breakpoints, and therefore these be the same piecewise linear function.
\end{proof}

\subsection{Finding the Best Piecewise Approximation}
As we did for symmetric splitting functions, we can quickly find the best piecewise linear $(1+\varepsilon)$-approximation to a $k$-GSCB integer function $\vw$ using a greedy approach. 
We omit proof details, as they exactly mirror arguments provided for the symmetric case. The submodularity constraint $2 \vw(i) \geq \vw(i+1) + \vw(i-1)$ for $i \in \{0\} \cup [k]$ can be viewed as a discrete version of concavity, and will ensure that the piecewise linear function returned by such a procedure will also be nonnegative and concave. After obtaining the piecewise linear approximation, we can apply Lemma~\ref{lem:reversecontext} to reverse engineer a $k$-CG function of a small order that approximates $\vw$. We obtain the same asymptotic upper bound on the number of linear pieces needed to approximate $\vw$.
\begin{lemma}
	\label{lem:upperlog}
	Let $\vw$ be a $k$-GSCB integer function and $\varepsilon \geq 0$. There exists a $k$-CG function $\hat{\vw}$ of order $J = O(\frac{1}{\varepsilon}\log k)$ that satisfies $\vw(i) \leq \hat{\vw}(i) \leq (1+\varepsilon) \vw(i) $ for any $i \in \{0\} \cup [k]$.
\end{lemma}

\subsection{Approximating Cardinality-Based Sum of Submodular Functions}
Recall that $k$-CG functions correspond to combinations of ACB-gadgets for a hyperedge $e$ as well as directed edges between nodes in $e$ and the source and sink nodes in some minimum $s$-$t$ cut problem. Each ACB-gadget involves one new auxiliary node and $2 |e|$ directed edges, and the number of ACB-gadgets is equal to the order of the $k$-CG function (the number of linear pieces minus one). Let $\mathcal{H} = (V, \mathcal{E})$ be a hypergraph with $n = |V|$ nodes, where each splitting function is submodular, cardinality-based, and is not required to be symmetric or penalize only cut hyperedges. Finding the minimum cut in $\mathcal{H}$ corresponds to solving the sum of submodular splitting functions given in~\eqref{soscard}. For $\varepsilon \geq 0$, we can preserve cuts in $\mathcal{H}$ to within a factor $(1+\varepsilon)$ by introducing a source and sink node $s$ and $t$ and applying our sparse reduction techniques to each hyperedge to obtain a directed graph $\hat{G} = (V \cup \mathcal{A} \cup \{s,t\}, \hat{E})$, where $\mathcal{A}$ is the set of auxiliary nodes, with $N = O(n + \frac{1}{\varepsilon}\sum_{e \in \mathcal{E}} \log |e|)$ nodes and $M = O(n+\frac{1}{\varepsilon}\sum_{e \in \mathcal{E}} \log |e|)$ edges. Even if the size of each $e \in \mathcal{E}$ is $O(n)$, we have $N = O(n + \varepsilon^{-1} |\mathcal{E}| \log n)$ and $M = O(\varepsilon^{-1} \log n \sum_{e \in \mathcal{E} }|e|)$.

\section{Parameter Settings for Image Segmentation Experiments}
\label{app:dsfm}
The continuous optimization methods for DSFM that we compare against in image segmentation experiments are implemented in C++ with a MATLAB front end. The Incidence Relation AP (IAP) method is an improved version of the AP method~\cite{nishihara2014convergence}. Li and Milenkovic~\cite{panli2018revisiting} showed that the runtime of the method can be significantly faster if one accounts for so-called \emph{incidence relations}, which describe sets of nodes that define the support of a component function. In our experiments we also ran the standard AP algorithm as implemented by Li and Milenkovic, but this always performed noticeably worse that IAP in practice, so we only report results for IAP. Neither of these methods require setting any hyperparameters.

Li and Milenkovic~\cite{panli2018revisiting} also showed that accounting for incidence relations leads to improved \emph{parallel} runtimes for ACDM and RCDM, but this does not improve serial runtimes. To simulate improved parallel runtimes, the implementations ACDM and RCDM of these authors include a parallelization parameter $\alpha = K/R$, where $K$ is the number of projections performed in an inner loop of these methods, and $R$ is the number of component functions. In theory, the $K$ projections could be performed in parallel, leading to faster overall runtimes. The comparative parallel performance between methods can be simulated by seeing how quickly the methods converge in terms of the number of total projections performed. Note however that the implementations themselves are serial, and only simulate what could happen in a parallel setting.

In our experiments our goal is to obtain the fastest possible serial runtimes. Li and Milenkovic~\cite{panli2018revisiting} demonstrated that the minimum number of total projections needed to achieve convergence to within a small tolerance is typically achieved when $\alpha$ is quite small. When projections are performed in parallel, choosing a larger $\alpha$ may still be advantageous. However, since our goal is to obtain the fast serial runtimes, we chose a small value $\alpha = 0.01$, based on the results of Li and Milenkovic. We also tried larger and smaller values of $\alpha$ in post-hoc experiments on all four instances of DSFM, though this led to little variation in performance. 

In addition to $\alpha$, ACDM relies on an empirical parameter $c$ controlling the number of iterations in an outer loop. We used the recommended default parameter $c = 10$. In general it is unclear how to set this parameter a priori to obtain better than default behavior on a given DSFM instance.
In order to highlight the strength of \textsc{SparseCard} relative to ACDM, we additionally tried post-hoc tuning of $c$ on each dataset to see how much this could affect results. We ran ACDM on each of the four instances of DSFM (2 image datasets $\times$ 2 superpixel segmentations each) for all $c \in \{10,25,50,100,200\}$, three different times, for $50R$ projections (i.e., on average we visit and perform a projection step at each component function 50 times). This took roughly 30-45 seconds for each run. We then computed the average duality gap for each $c$ and each instance over the three trials, and re-ran the algorithm for even longer using the best value of $c$ for each instance. The result is shown as ACDM-best in Figure~\ref{fig:images}. \textsc{SparseCard} still maintains a clear advantage over this method on the instances that involve 200 superpixels (i.e., very large region potentials). Our method also obtains better approximations for the \emph{smallplant} dataset with 500 superpixels for the first 40 seconds, after which point ACDM-best converges to the optimal solution. Nevertheless, we would not have been able to see this improved behavior without post-hoc tuning of the hyperparameter $c$ for ACDM. Meanwhile, \textsc{SparseCard} does not rely on any parameter except $\varepsilon$, and it is very easy to understand how setting this parameter affects the algorithm as it directly controls the sparsity and a priori approximation guarantee for our method.

We remark finally that Li and Milenkovic~\cite{panli2018revisiting} also considered and implemented an alternate version of RCDM (RCDM-greedy) with a greedy sampling strategy for visiting component functions in the method's inner loop. Despite being advantageous for parallel implementations, we found that in practice that this method had worse serial runtimes, so we did not report results for it.

\bibliographystyle{plain}
\bibliography{main-arXiv-v2.bib}
\end{document}